\documentclass[a4paper,onecolumn,11pt,accepted=2025-10-24]{quantumarticle}
\pdfoutput=1
\usepackage[utf8]{inputenc}
\usepackage[english]{babel}
\usepackage[T1]{fontenc}
\usepackage{amsmath}
\usepackage{amsthm}
\usepackage{hyperref}
\usepackage{tikz}
\usepackage{amsfonts,amssymb,amsthm, bbm,braket,mathtools}
\usepackage{mathtools}
\usepackage{graphicx}   
\usepackage{amsbsy} 
\usepackage{comment}
\usepackage{algorithm}
\usepackage{tikz}
\usepackage{cancel}
\usepackage{silence}
\usepackage{bbm}
\usepackage{subcaption}

\newtheorem{theorem}{Theorem}[section]
\newtheorem*{theorem*}{Theorem}
\newtheorem{problem}[theorem]{Problem}
\newtheorem{lemma}[theorem]{Lemma}
\newtheorem{fact}[theorem]{Fact}
\newtheorem*{lemma*}{Lemma}
\newtheorem{corollary}[theorem]{Corollary}
\newtheorem{proposition}[theorem]{Proposition}

\newtheorem*{remark}{Remark}

\newtheorem{definition}[theorem]{Definition}
 
\usepackage{caption}
\usepackage{subcaption}
\usepackage{graphicx}
\usepackage[utf8]{inputenc}
\usepackage{physics}
\usepackage{algpseudocode}
\usepackage{geometry}

\begin{document}

\title{Complexity of Fermionic 2-SAT}

\author{Maarten Stroeks}
\affiliation{Delft Institute of Applied Mathematics, Delft University of Technology, The Netherlands}
\affiliation{QuTech, Delft University of Technology, The Netherlands}
\email{m.e.h.m.stroeks@tudelft.nl}
\author{Barbara M. Terhal}
\affiliation{Delft Institute of Applied Mathematics, Delft University of Technology, The Netherlands}
\affiliation{QuTech, Delft University of Technology, The Netherlands}


\maketitle

\begin{abstract}
We introduce the fermionic satisfiability problem, \textsc{Fermionic} $k$\textsc{-SAT}: this is the problem of deciding whether there is a fermionic state in the null-space of a collection of fermionic, parity-conserving, projectors on $n$ fermionic modes, where each fermionic projector involves at most $k$ fermionic modes. We prove that this problem can be solved efficiently classically for $k=2$. In addition, we show that deciding whether there exists a satisfying assignment with a given fixed particle number parity can also be done efficiently classically for \textsc{Fermionic} 2-SAT: this problem is a quantum-fermionic extension of asking whether a classical 2-SAT problem has a solution with a given Hamming weight parity. We also prove that deciding whether there exists a satisfying assignment for particle-number-conserving \textsc{Fermionic} 2-SAT for some given particle number is NP-complete.
Complementary to this, we show that \textsc{Fermionic} $9$-SAT is ${\rm QMA}_1$-hard.
\end{abstract}

\tableofcontents

\section{Introduction}
Classical $k$-satisfiability ($k$-SAT) is the problem of deciding whether there exists a configuration of $n$ Boolean variables which satisfies a collection of $k$-variable Boolean clauses. While $2$-SAT is solvable in linear time \cite{Lintimeclassical2sat}, $(k\geq 3)$-SAT is NP-complete. Bravyi introduced the problem \textsc{Quantum} $k$-SAT \cite{bravyi:quantumSAT}: this is the problem of deciding whether there exists an $n$-qubit state which is in the null-space of some collection of $k$-qubit projectors. Bravyi showed that \textsc{Quantum} $2$-SAT can be decided classically efficiently, and later it was shown it can even be done in linear time \cite{ASSZ:linearSAT, lineartimeQ2sat2}. For $k\geq 3$, it has been shown that \textsc{Quantum} $k$-SAT is QMA$_{1}$-complete \cite{bravyi:quantumSAT, GN:3SAT}. 

Here, we introduce the \textit{fermionic} satisfiability problem. An $n$-mode fermionic system corresponds to set of annihilation operators $\{a_j\}_{j=1}^n$ (with Hermitian conjugate creation operators $\{a_{j}^{\dagger}\}_{j=1}^n$) which satisfy $\{a_i^{\dagger},a_j\} := a_i^{\dagger} a_j+a_j a_i^{\dagger} =\delta_{ij}I$, $\{a_i,a_j\}=0$. In addition, there is a state $\ket{\rm vac}$, ---called the vacuum, or empty, zero-particle state---, for which $a_i \ket{\rm vac}=0, \; \forall i$. 

\textsc{Fermionic} $k$-SAT is the problem of deciding whether there exists an $n$-mode fermionic state which is in the null-space of a collection of $k$-mode fermionic projectors. Each $k$-mode projector is a polynomial in $a_{j}$ and $a_{j}^{\dagger}$ with $j\in S$, where $S\subseteq [n]$ and $|S| = k$. Furthermore, the $k$-mode projector commutes with the parity operator $\hat{P}:=(-1)^{\sum_{i=1}^n a_i^{\dagger} a_i}$.
As a special case, we also study \textsc{Particle-number-conserving (PNC) Fermionic} $k$\textsc{-SAT}, where each projector additionally commutes with the particle-number operator $\hat{N}:=\sum_{i=1}^n a_i^{\dagger} a_i$. In addition to studying the satisfiability of these problems, we investigate the complexity of \textsc{PNC Fermionic} $2$\textsc{-SAT} and \textsc{Fermionic} $2$\textsc{-SAT} when, respectively, the particle number is fixed to some given value $\hat{N}=N\in \{0,1\ldots, n\}$, or the parity $\hat{P}=P\in \{+1,-1\}$ is fixed.

\subsection{Main results}
\label{sec:mainresults}
\begin{problem}[\textsc{Fermionic} $k$\textsc{-SAT}]
Given is a set of projectors $\{\Pi_{i}\}$ on $n$ fermionic modes, where each $\Pi_{i}$ is a polynomial in the operators $a_{j}^{\dagger}$ and $a_{j}$ with index 
$j$ contained in a subset of at most $k$ fermionic modes. Furthermore, $[\Pi_i,\hat{P}] = 0$ for all projectors $\{\Pi_i\}$ with $\hat{P}=(-1)^{\sum_{i=1}^n a_i^{\dagger} a_i}$.
Decide whether there exists a $\ket{\psi}$ s.t. $\forall i, \Pi_i\ket{\psi} = 0$  (YES), or for all $\ket{\psi},\;\sum_i\bra{\psi}\Pi_i\ket{\psi} > \frac{1}{{\rm poly}(n)}$ (NO).
\label{problem:pcFkSAT}
\end{problem}

To be precise in our definition of \textsc{Fermionic} $k$\textsc{-SAT} one should, in principle, provide the accuracy with which the projectors $\Pi_i$ are specified, i.e.~with respect to what elementary gate set, see e.g. \cite{ASSZ:linearSAT, rudolph_precise} for how this is handled for \textsc{Quantum} $k$-\textsc{SAT}. Since this is separate from our results, we leave such precise definition to future work.

We also consider \textsc{Fermionic} $k$-SAT with global additional constraints, i.e. one asks for an assignment with given particle number or parity:

\begin{problem}[\textsc{Fermionic} $k$\textsc{-SAT} with fixed parity $P$]
Given $n$ fermionic modes, $P\in \{-1,+1\}$, and a set of projectors $\{\Pi_i\}$ where each projector $\Pi_i$ is a polynomial in the operators $a_{j}^{\dagger}$ and $a_{j}$ with index $j$ contained in a subset of at most $k$ fermionic modes. Furthermore, $\forall i, [\Pi_i,\hat{P}] = 0$. Decide whether there exists $\ket{\psi}$ s.t. $\forall i, \;\Pi_i\ket{\psi} = 0$ and $\hat{P}\ket{\psi} = P\ket{\psi}$ (YES), or 
$\forall \ket{\psi}$ s.t. $\hat{P}\ket{\psi} = P\ket{\psi}$, $\sum_i\bra{\psi}\Pi_i\ket{\psi} > \frac{1}{{\rm poly}(n)}$ (NO).
\label{problem:parityconstrainedFkSAT}
\end{problem}

A special case of Problem \ref{problem:pcFkSAT} is the case when each projector commutes with the particle number, i.e. $\forall i,\;[\Pi_i,\hat{N}] = 0$, which we will refer to as \textsc{Particle-number-conserving (PNC) Fermionic} $k$\textsc{-SAT}. We also consider the following problem:

\begin{problem}[\textsc{PNC Fermionic} $k$\textsc{-SAT} with fixed particle number $N$]
Given $n$ fermionic modes, an integer $N\in \{0,1,\ldots,n\}$ and a set of projectors $\{\Pi_i\}$ where each projector $\Pi_i$ is a polynomial in the operators $a_{j}^{\dagger}$ and $a_{j}$ with index $j$ contained in a subset of at most $k$ fermionic modes. Furthermore, $\forall i, [\Pi_i,\hat{N}] = 0$. Decide whether there exists a state $\ket{\psi}$ s.t. $\forall i,\;\Pi_i\ket{\psi} = 0$ and $\hat{N}\ket{\psi} = N\ket{\psi}$ (YES), or 
$\forall \ket{\psi}$ s.t. $\hat{N}\ket{\psi} = N\ket{\psi}$, $\sum_i\bra{\psi}\Pi_i\ket{\psi} > \frac{1}{{\rm poly}(n)}$ (NO).
\label{problem:NparticleFkSAT}
\end{problem}

\bigbreak 
\noindent
Our main results are:

\begin{theorem}
\normalfont{\textsc{Fermionic \normalfont{2}-SAT}} \textit{$\in$} \normalfont{\textsc{P}}, \textit{and can be decided in time $O(n+m)$, where $m=|\{\Pi_{i}\}|$ denotes the number of projectors.}
\label{theorem:parityconservingP}
\end{theorem}

\begin{theorem}
\normalfont{\textsc{Fermionic \normalfont{2}-SAT} with fixed parity} \textit{$\in$} \normalfont{\textsc{P}}, \textit{and can be decided in time $O(nm)$, where $m=|\{\Pi_{i}\}|$ denotes the number of projectors.} 
\label{theorem:parityconstrainedP}
\end{theorem}

\begin{theorem}
\normalfont{\textsc{PNC Fermionic \normalfont{2}-SAT} with fixed particle number} \textit{is} \normalfont{\textsc{NP}}\textit{-complete}. 
\label{theorem:NPcomplete}
\end{theorem}

As part of the proof of Theorem \ref{theorem:parityconstrainedP}, we show that classical 2-SAT with fixed (Hamming weight) parity can be solved in $O(nm)$ time, see Theorem \ref{thm:parity} in Appendix \ref{app:parityconstrained2SAT} for details. To the best of our knowledge, this has not been proved before.

To provide further context for our results, we also provide a straightforward proof that \textsc{Fermionic} $k$-\textsc{SAT} is \textsc{QMA}$_1$-hard for $k=9$ in Section \ref{sec:kSAT}.

\subsection{Overview of paper and main ideas}

\textsc{Fermionic} $2$-\textsc{SAT} instances are defined on a graph $G$ whose vertices are fermionic modes and whose edges correspond to projectors involving pairs of modes. These projectors are of rank at most three (if there is to be a satisfying assignment) in the two fermionic mode subspace, consisting of at most three rank-1 projectors. These rank-1 projectors --- we will refer to them as \textit{clauses} --- can be either genuinely quantum clauses, or classical clauses which exclude a particular occupation of the two modes. When {\em restricted to a two-qubit space}, a genuinely quantum clause can exclude support on a state such as $\alpha\ket{00}+\beta \ket{11}$ (with $|\alpha|,|\beta| \neq 0$) which we call a $\Pi^{02,q}_e$ clause, or exclude support on a state such as $\alpha\ket{01}+\beta \ket{10}$ (with $|\alpha|,|\beta| \neq 0$) which we call a $\Pi^{1,q}_e$ clause. Since projectors in \textsc{Fermionic} $2$-\textsc{SAT} are parity conserving by definition, these are the only allowed quantum clauses besides the classical clauses. Note that we mention the caveat {\em when restricted to a two-qubit space} since fermionic operators when viewed as acting on an $n$-qubit space are generally non-local due to Jordan-Wigner strings. 

Let us discuss the ideas behind the proofs of our main results of Section \ref{sec:mainresults}. Given the graph $G$, we first remove all classical clauses, leaving us possibly with a set of disconnected graphs which we call quantum clusters, see e.g. Figure \ref{fig:quantumclusters} in Section \ref{sec:qclusters}. We will refer to modes that are not involved in any quantum clauses as classical modes. 
It turns out to be useful to distinguish at least two types of quantum clusters, namely ones that are so-called hidden particle-number-conserving (hPNC) and ones which are not (non-hPNC). The hPNC clusters can be brought to a form in which they only contain $\Pi^{1,q}_e$ quantum clauses by a particle-hole transformation defined in Section \ref{sec:ph-trafo}: after this transformation, every clause commutes with the cluster particle number operator, hence the terminology. In Section \ref{section:prelims} we introduce and discuss these various concepts.

In Section \ref{sec:eff} we start by examining the effect of `constraint propagation' by $\Pi^{1,q}_e$ and $\Pi^{02,q}_e$ quantum clauses, see Lemma \ref{lemma:partners} and its Corollary \ref{cor:allparticlenumberallparity}. This leads up to a characterization of the satisfying assignments: wlog, they are of --- what we call --- \textit{cluster-product form}, see Proposition \ref{lemma:global}. Compared to \textsc{Quantum} 2-SAT which, except for on some rank-3 projectors, always has a product satisfying assignment (if any) \cite{ASSZ:linearSAT}, this form is more general, but also restricted in an interesting way. For assignments with or without overall fixed parity, we will show that one can limit oneself to cluster-product states which are products of operators (acting on the vacuum) which create (1) classical occupations on so-called classical modes, (2) Gaussian states on hPNC clusters, (3) Gaussian states on non-hPNC clusters which are lines or loops, and (4) non-Gaussian states on certain 4-fermion non-hPNC clusters. 
To obtain this result, we proceed in Section \ref{sec:ferm-qubit} to exclude certain quantum numbers, ---these are hidden cluster particle number or cluster parity which uniquely fix a cluster satisfying assignment---, on quantum clusters. Crucial in this is Lemma \ref{lemma:degreeatleast3} which provides a uniquely fermionic simplification, which leaves us to analyze quantum clusters which are lines and loops in Section \ref{sec:linesandloops}. Fermions on a line do map directly onto qubits, hence we can use some proof techniques borrowed from \textsc{Quantum} 2-SAT investigations.
Importantly, these results allow one to efficiently (with effort linear in the number of modes and clauses) verify which quantum numbers, i.e. hidden cluster particle number and cluster parity, are allowed on quantum clusters in the cluster-product form. This is expressed in Lemma \ref{lem:ccheck}. 

To prove Theorems \ref{theorem:parityconservingP} and \ref{theorem:parityconstrainedP} in Section \ref{sec:eff}, we bring back all classical clauses in the graph $G$ and reduce the problem of deciding \textsc{Fermionic} 2-SAT after several preprocessing steps to deciding a classical 2-SAT instance (thereby proving Theorem \ref{theorem:parityconservingP}). For Theorem \ref{theorem:parityconstrainedP}, there is the additional constraint of asking for a satisfying assignment of fixed parity. In Appendix \ref{app:parityconstrained2SAT}, we provide an efficient algorithm for deciding classical 2-SAT where one asks for a solution of fixed Hamming weight parity (Theorem \ref{thm:parity}), which is used as a subroutine in the algorithm for the fermionic problem (Theorem \ref{theorem:parityconstrainedP}). This result on parity-constrained classical 2-SAT may not be surprising, but we are unaware of any previous algorithm or proof. Crucial in the reductions in the proofs of Theorems \ref{theorem:parityconservingP} and \ref{theorem:parityconstrainedP} is the fact that non-classical assignments on quantum clusters heavily constrain how classical clauses which straddle the quantum cluster and a classical mode can be satisfied, namely the clause has to be true by only the choice of the classical mode, since the non-classical assignment requires that each mode inside the cluster should be allowed to be both empty (0) \textit{and} filled (1). The nature of the hPNC and non-hPNC clusters asks for a different treatment in the proofs of these theorems. 

In Section \ref{sec:NPC}, we examine the complexity of particle-number-conserving \textsc{Fermionic} 2-SAT where we ask for an assignment with given fixed particle number, and prove Theorem \ref{theorem:NPcomplete}. Given the efficiently checkable cluster-product form, it may not be surprising that this problem is contained in NP. It can be proved to be NP-complete since the class of problems contains the weighted 2-SAT problem which asks for a 2-SAT satisfying assignment with fixed Hamming weight, as a special case. In Section \ref{sec:kSAT}, we prove a quantum hardness result for \textsc{Fermionic} $k$-SAT, still leaving a considerable gap between our efficient classical algorithms for $k=2$ and this hardness result for $k=9$. We discuss open problems in the Discussion Section \ref{sec:discuss}.

\subsubsection{Differences with Quantum 2-SAT}
We take a paragraph to reflect on the difference between our \textsc{Fermionic} 2-SAT findings and \textsc{Quantum} 2-SAT. 

For \textsc{Quantum} $k$-SAT and \textsc{Fermionic} $k$-SAT, the projectors act nontrivially on $k$ qubits resp. involve $k$ fermionic modes. Fermion-to-qubit mappings allow for the projectors of a given instance of \textsc{Fermionic} $k$-SAT to be mapped to qubit projectors. However, each fermionic projector involving $k$ modes gets mapped ---in a general setting--- onto a $q$-qubit projector, where $q\geq k$ \cite{bravyi00ferm}. Therefore, one cannot deduce the complexity of \textsc{Fermionic} $k$-SAT directly from the complexity of \textsc{Quantum} $k$-SAT. 

Central to the solution of \textsc{Quantum} $2$-SAT is the fact that a satisfying assignment, if it exists, can be taken to be a product state (except for two-qubit projectors whose null-space is a single entangled state), see e.g. Theorem 4 in \cite{ASSZ:linearSAT}. We do not have such a statement for \textsc{Fermionic} 2-SAT. In fact, we find that satisfying assignments on non-hPNC clusters can include small 4-fermion non-Gaussian states. Furthermore, for \textsc{Fermionic} $2$-SAT, quantum clusters with vertices with degree at least 3 obey specific constraints which relate to uniquely fermionic signs, see Lemma \ref{lemma:degreeatleast3} in Section \ref{sec:ferm-qubit}, excluding many particle numbers (but leaving only Gaussian states for hPNC clusters). In contrast, for quantum 2-SAT, product state assignments are generally built from superpositions of all particle numbers.
Indirectly, Lemma \ref{lemma:degreeatleast3} also limits the use of any possible non-Gaussian assignment: non-Gaussian states are only \textit{sometimes} needed on non-hPNC clusters of at most 4 fermionic modes (and any state of three or fewer fermionic modes is always Gaussian \cite{MCT}). In Appendix \ref{app:exampleNG} we work through this 4-fermion non-Gaussian example explicitly, showing how the departure from product state assignments for \textsc{Quantum} 2-SAT comes about.


\section{Preliminaries} 
\label{section:prelims}

In this section we review fermionic language and concepts, and introduce several tools for the treatment of assignments on quantum clusters.


For a subset $S\subseteq [n]$ and $|S|=k$, we  define the operator $a_S^{\dagger} := a_{i_{1}}^{\dagger}a_{i_{2}}^{\dagger}\ldots a_{i_{k}}^{\dagger}$, where $i_{1}<i_{2}<\ldots <i_{k}$. This ordering of the fermionic modes will later be conveniently chosen depending on the \textsc{Fermionic} 2-SAT instance in Section \ref{sec:qclusters}. We can write any $n$-mode fermionic state as 
\begin{align}
   \ket{\psi} =  \sum_{S\subseteq [n]}\alpha_{S}\ket{S}, \; \ket{S} := a_S^{\dagger}\ket{\rm vac},
\end{align} 
where $\alpha_S \in \mathbb{C}$ is the amplitude for state $\ket{S}$. Each state $\ket{S}$ corresponds to a configuration of mode {\em occupations} $\{x_i\}_{i=1}^n$, i.e. the bit $x_i=1$ (resp. $x_i=0$) signifies that mode $i$ is occupied, so $i\in S$, (resp. unoccupied, so $i \notin S$). When we speak of a \textit{classical assignment} in the remainder of this work, we mean a state of type $\ket{S} = a_{S}^{\dagger}\ket{\rm vac}$.  


One can map an (ordered) set of fermionic creation and annihilation operators via the Jordan-Wigner (JW) transformation onto a set of qubit operators: \begin{align}
        a_i^{\dagger} \stackrel{J.W.}{=} Z_1\ldots Z_{i-1}\sigma_i^+ ,
        \label{eq:JW}
    \end{align} and the fermionic vacuum state $\ket{\rm vac} \rightarrow \ket{00\ldots 0}$, using the definition $\sigma_i^+=\ket{1}\bra{0}_i$ (and similarly $\sigma_i^-=\ket{0}\bra{1}_i$).
    
Gaussian fermionic states are thermal states and eigenstates of \textit{quadratic} fermionic Hamiltonians --- corresponding to so-called non-interacting fermions --- and these states can be fully characterized by a $2n \times 2n$ correlation matrix 
\begin{align}
    M_{ij}=\frac{i}{2}{\rm Tr}([c_i, c_j] \rho),\; c_{2i-1}=a_i+a_i^{\dagger}, \; c_{2i}=i(a_i-a_i^{\dagger}),\;  \{c_i, c_j\}=2\delta_{ij}I.
    \label{eq:cor-mat}
\end{align}
 A state is a pure Gaussian state iff its correlation matrix obeys $M^T M=I$, see e.g. \cite{MCT, Kraus_2009}. The $n$-mode states with particle number $\hat{N}=0$ (vacuum), $\hat{N}=1$ (single-particle "Slater determinant" state), $\hat{N}=n-1$ (single-hole state) and $\hat{N}=n$ (all-filled) are Gaussian states. The classical states $\ket{S}$ are also examples of Gaussian states. See more information on Gaussian states in \cite{Surace_2022}. 

\subsection{Characterization of the Fermionic 2-SAT projectors} 
\label{sec:char}
Let $G=(V,E)$ be a graph such that the vertices $v \in V$ label the modes ($|V| = n$) and each edge is associated with a projector $\{\Pi_e\}_{e\in E}$. Since the projectors in \textsc{Fermionic} 2\textsc{-SAT} are parity conserving, any projector $\Pi_{e = (j,k)}$ is a polynomial in $a_{j},a_{j}^{\dagger}, a_{k}$ and $a_{k}^{\dagger}$ of even degree.
Obviously, since each projector $\Pi_e$ on modes $j$ and $k$ commutes with the overall parity operator $\hat{P}$, it also commutes with the edge parity $P_{jk}=(-1)^{a_j^{\dagger}a_j+a_k^{\dagger}a_k}$. Similarly, if $\Pi_e$ commutes with $\hat{N}$, it commutes with the edge particle number $\hat{N}_{jk}=\hat{N}_j+\hat{N}_k$. Hence, if we consider a two-mode projector in isolation, it can be represented as a \textsc{Quantum} 2-SAT projector on two qubits which is parity-conserving (or, in special cases, also particle number-conserving). This means that if the projector, say of rank-1 in the 2-qubit space, projects onto a product state, the product state can only be a computational basis state. And if it projects onto an entangled state, then in the number-conserving case, it can only project onto states such like $\alpha \ket{01}+\beta \ket{10}$. In the merely parity-conserving case, it can also project onto a state $\alpha \ket{00}+\beta \ket{11}$ for some $\alpha, \beta$. 
This picture is entirely preserved when we consider the action of the fermionic projectors onto $n$ modes. In principle, one can apply a Jordan-Wigner transformation, see Eq.~\eqref{eq:JW}, and represent these fermionic projectors as 2-qubit projectors to which strings of Pauli $Z$'s are appended, so that the qubit projectors are not 2-local, but this does not affect the number- or parity-conserving properties. In some sense, we thus consider a simplified form of \textsc{Quantum} 2-SAT due to number- and parity-conserving constraints, but the fermionic nature of the problem, i.e. the effect of the Pauli $Z$-strings through the Jordan-Wigner transformation, makes the problem genuinely different from \textsc{Quantum} 2-SAT.

Let's now discuss the nature of rank-1 projectors\footnote{They have rank-1 only in the 2-mode subspace, not in the $n$-mode fermionic space, but for convenience we use this term throughout this paper. Same for statements on higher rank projectors which have a rank $r\leq 4$ in the 2-mode subspace.} in fermionic language. With the two-qubit picture in mind, we have two different types of rank-1 projectors for an edge $e=(j,k)$ (with convention $j < k$). The only possible rank-1 projector onto an odd-parity state is
\begin{multline}
\Pi^1_{e}=
|\beta_{e}|^2a_{j}^{\dagger}a_{j}a_{k}a_{k}^{\dagger}+\beta_{e} \gamma_{e}^* a_{j}^{\dagger}a_{k} - \gamma_{e} \beta_{e}^{*}a_{j}a_{k}^{\dagger} + |\gamma_{e}|^{2}a_{j}a_{j}^{\dagger}a_{k}^{\dagger}a_{k} \\ \quad\rightarrow\: \;\mbox{"project onto/exclude $\beta_e\ket{10}+\gamma_e\ket{01}$"},
\label{eq:projectordef_1}
\end{multline}
with $|\beta_e|^2+|\gamma_e|^2=1$. When $\Pi^1_e$ does not project onto just "01" or "10", we call this a genuinely quantum clause and write $\Pi^{1,q}_{e}$, otherwise we refer to them as classical clauses $\Pi^{1,c}_e$. The only type of rank-1 projector onto an even-parity state is $\Pi^{02}_{e}$: 
\begin{multline}
    \Pi^{02}_{e} = |\alpha_{e}|^2 a_{j}a_{j}^{\dagger}a_{k}a_{k}^{\dagger} - \alpha_{e}\delta_{e}^{*}a_{j}a_{k} + \delta_{e}\alpha_e^{*}a_{j}^{\dagger}a_{k}^{\dagger} + |\delta_{e}|^{2}a_{j}^{\dagger}a_{j}a_{k}^{\dagger}a_{k} \\ \quad\rightarrow\: \;\mbox{"project onto/exclude $\alpha_e\ket{00}+\delta_e\ket{11}$"},
    \label{eq:projectordef_02}
\end{multline}
with $|\alpha_e|^2+|\delta_e|^2=1$. Again, when $\Pi^{02}_{e}$ does not project onto just "00" or "11", it is a genuinely quantum clause and we write $\Pi^{02,q}_{e}$, otherwise it is a classical projector $\Pi^{02,c}_{e}$. 

For convenience, we introduce the following shorthand for a given quantum clause on an edge $e$:
\begin{align}
u_{e} :=&\: \gamma_{e}^{*}/\beta_{e}^{*} \in \mathbb{C},\text{ for }\Pi^{1,q}_{e},\nonumber \\
v_{e} :=&\: \alpha_{e}^{*}/\delta_{e}^{*} \in \mathbb{C},\text{ for }\Pi^{02,q}_{e},
\label{eq:defue}
\end{align}
which we shall only use later on when $\beta_{e}, \delta_{e} \neq 0$, i.e. in Lemmas \ref{lemma:partners}, \ref{lemma:degreeatleast3}, \ref{lem:loop} and \ref{lem:nonhPNClinesloops}.


For PNC \textsc{Fermionic} 2-SAT, we only have projectors of type $\Pi^{1}_{e}$ in Eq. 
\eqref{eq:projectordef_1}, and only classical projectors 
\begin{align}
\Pi^0_{e}= &\: a_{j}a_{j}^{\dagger}a_{k}^{}a_{k}^{\dagger} \quad\rightarrow\: \;\mbox{"project onto/exclude $\ket{00}$"}, \nonumber \\ 
\Pi^2_{e}= &\: a_{j}^{\dagger}a_{j} a_{k}^{\dagger}a_{k} \quad\rightarrow\: \;\mbox{"project onto/exclude $\ket{11}$"}.
\label{eq:pi0andpi2}
\end{align}
Proposition \ref{prop:02invariance} in Appendix \ref{app:prelims} proves that these projectors are invariant under taking any unitary linear combination of the two modes $j$ and $k$, but this understanding also follows from the two-qubit picture above.

\subsubsection{Higher-rank projectors}
\label{sec:higherrank}
The projectors appearing in \textsc{Fermionic} 2\textsc{-SAT} are of rank at most three and can be obtained by summing the unit-rank projectors above as long as these project onto orthogonal states. This is true since the $1$-eigenvalue sub-space of a projector $\Pi_{e}$ is spanned by eigenstates of the commuting edge parity operator $\hat{P}_{e=(j,k)}$, or, in the particle-number-conserving case, by eigenstates of the commuting edge number operator $\hat{N}_{e=(j,k)}$. By viewing the action of the rank-1 projectors in the two-qubit space, it is clear that there exist two orthogonal projectors in the $\{\ket{01,10}\}$ space and two orthogonal projectors in the $\{\ket{00},\ket{11}\}$ space. 

Observe that a rank-2 projector which is the sum of two orthogonal $\Pi^{02,q}_e$ clauses is a rank-2 classical projector (a sum of two classical clauses, excluding both "00" and "11"), and similarly a rank-2 projector which is the sum of two orthogonal $\Pi_e^{1,q}$ clauses is classical (excluding both "01" and "10"). 
We will call an edge {\em a classical edge} ($\in E_c$) when its at-most-rank-3 projector is constructed from rank-1 classical clauses, otherwise we call it a quantum edge ($\in E_q$), thus the edgeset can be written as $E=E_c \cup E_q$. 
Hence, an edge in $E_q$ consists of one --- $\Pi^{02,q}_e$ \textit{or} $\Pi^{1,q}_e$ --- or two --- $\Pi^{02,q}_e$ \textit{and} $\Pi^{1,q}_e$ --- quantum clauses. Otherwise, the edge belongs to $E_c$.
It is worthwhile to note that a rank-2 projector $\Pi_e^{1,q}+\Pi_e^{02,q}$ only consists of quadratic terms in creation and annihilation operators, i.e.~the quartic contributions cancel as the quartic contribution $\propto a_j^{\dagger}a_j a_k^{\dagger}a_k$ of each quantum clause is independent of the choice of $\alpha_e,\beta_e,\gamma_e,\delta_e$ and opposite in sign for a $\Pi^{02,q}_e$ clause versus a $\Pi^{1,q}_e$ clause as can be seen from Eqs.~\ref{eq:projectordef_1} and \ref{eq:projectordef_02}. This fact will be used in Lemma \ref{cor:onlygaussianSAs}. 

\subsection{Particle-hole transformation}
\label{sec:ph-trafo}
We define the unitary particle-hole transformation $K_{S}$ on a subset $S \subseteq [n]$ of the $n$ modes. The transformation interchanges the creation of a particle in $S$ with the annihilation of a particle, i.e. the creation of a hole or absence of a particle, in $S$:
\begin{align}
    \forall j\in S \quad \rightarrow&\: \quad K_{S}a_{j}K_{S}^{-1} = \bigg[\hspace{-0.1cm}\prod_{\substack{k<j \\ \text{ s.t. } k\in S}} \hspace{-0.3cm} (-1) \hspace{0.1cm}\bigg] \hspace{0.1cm} a_{j}^{\dagger}, \hspace{0.5cm} K_{S}a_{j}^{\dagger}K_{S}^{-1} = \bigg[\hspace{-0.1cm}\prod_{\substack{k<j \\ \text{ s.t. } k\in S}} \hspace{-0.3cm} (-1) \hspace{0.1cm}\bigg] \hspace{0.1cm} a_{j}, \nonumber \\
    \forall j\notin S \quad \rightarrow&\: \quad K_{S}a_{j}K_{S}^{-1} = \bigg[\hspace{-0.1cm}\prod_{\substack{k<j \\ \text{ s.t. } k\in S}} \hspace{-0.3cm} (-1) \hspace{0.1cm}\bigg] \hspace{0.1cm} a_{j}, \hspace{0.5cm} K_{S}a_{j}^{\dagger}K_{S}^{-1} = \bigg[\hspace{-0.1cm}\prod_{\substack{k<j \\ \text{ s.t. } k\in S}} \hspace{-0.3cm} (-1) \hspace{0.1cm}\bigg] \hspace{0.1cm} a_{j}^{\dagger}.
\label{eq:ph_transformation}
\end{align}
$K_S$ also transforms the vacuum state $\ket{{\rm vac}}$, i.e. $\ket{\rm vac} \rightarrow K_{S}\ket{\rm vac}$, 
where the new vacuum state corresponds to the all-filled state for the modes in the set $S$, and the original vacuum for the modes which are not in the set $S$. The $\pm$ signs in the transformation $K_{S}$ in Eq.~\eqref{eq:ph_transformation} can be understood by performing a Jordan-Wigner transformation which leads to the manifestly-unitary transformation $K_{S} = \bigotimes_{k\in S} X_{k}$, with $X_{k}$ denoting Pauli-$X$ on qubit $k$. When $S=\emptyset$, $K_{S=\emptyset}=I$.

When we apply the particle-hole transformation $K_{S}$ to the genuinely quantum projectors in Eqs.~\eqref{eq:projectordef_1} and  \eqref{eq:projectordef_02}, it can interchange $\Pi_e^{1,q}$ clauses and $\Pi_e^{02,q}$ clauses, depending on the edge $e$ and the subset $S$. In Appendix \ref{app:prelims} we give the effect of the transformation in mathematical detail (which is not extremely insightful). 

An additional property of the particle-hole transformation $K_S$ is that it preserves Gaussianity of a state $\rho$. We can express the unitary transformation $K_S$ on the Majorana operators $c_{2i-1}\stackrel{J.W.}{=}Z_1\ldots Z_{i-1}X_i, c_{2i}\stackrel{J.W.}{=}Z_1\ldots Z_{i-1}Y_i$ (consistent with Eqs.~\ref{eq:JW} and ~\ref{eq:cor-mat}). Then $K_S$ is some product of Majorana operators $\Pi_{i \in S_{\rm Maj}} c_i$, using $X_k=\big[\prod_{1\leq j < k}(-ic_{2j-1}c_{2j})\big]c_{2k-1}$. We have $K_S c_{k} K_{S}^{-1}=\pm c_k$, with the sign depending on whether $k\in S_{\rm Maj}$ and $|S_{\rm Maj}|$, inducing a (diagonal) orthogonal transformation in the $\{c_i\}_{i=1}^{2n}$ basis. Any orthogonal transformation $R\in O(2n)$ preserves the Gaussianity of a pure fermionic state as it preserves the property that the correlation matrix $M$ still obeys $M^T M=I$ \cite{MCT}, hence $K_S$ preserves Gaussianity. Note that $K_S\hat{P} K_S^{-1}=\pm \hat{P}$, switching the overall parity $P$ when $|S_{\rm Maj}|$ is odd. 


\subsection{Quantum clusters}
\label{sec:qclusters}

The graph $G=(V,E)$ describing the fermionic 2-SAT problem (where an edge represents a rank-1, 2 or 3 projector) can contain zero, one or more {\em quantum clusters}, see Fig.~\ref{fig:quantumclusters}, defined as follows.

\begin{definition}[Quantum clusters]
Take the subgraph $G_{\rm sub}$ of $G$ obtained by removing all classical clauses in $E_c$, and removing all vertices that are not touching any $e\in E_{q}$. We define the quantum clusters $G_{q}^{i} = (V^{i}_{q},E^{i}_{q})$ (for $i=1,\ldots,Q$, and with $n_q^{i}:=|V^{i}_{q}|$) to be the connected components of $G_{\rm sub}$. The cluster particle number equals $\hat{N}^i_q = \sum_{j\in V^i_q} a_j^{\dagger} a_j$ and the cluster parity is $\hat{P}^i_q = (-1)^{\sum_{j\in V^i_q} a_j^{\dagger} a_j}$.
\label{def:quantumclusters}
\end{definition}

\begin{figure}[t]
\centering
\includegraphics[width=0.40\textwidth]{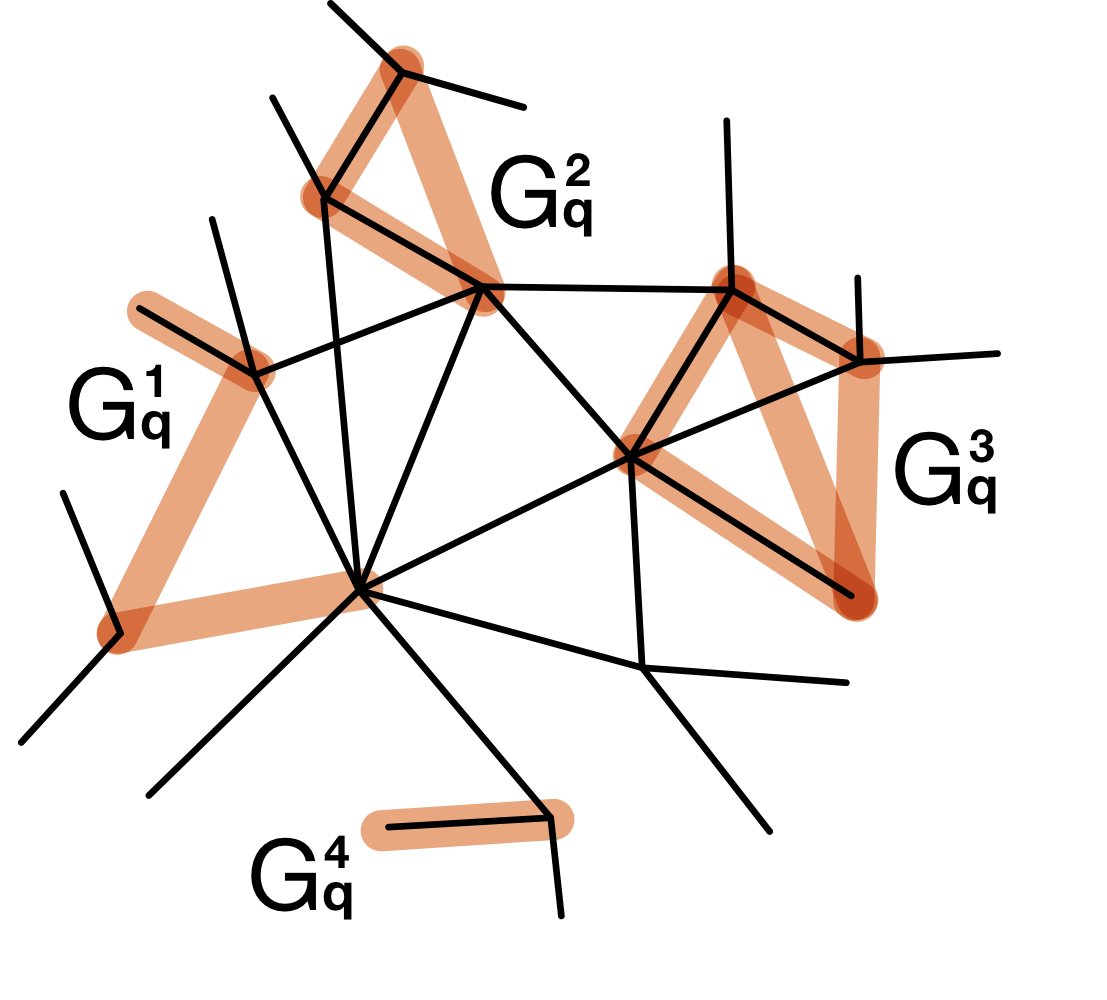}
\caption{Several connected subgraphs $G_q^i$ of $G$ with genuinely quantum clauses (orange), the other clauses (black) are classical clauses. The graph $G_{\rm sub}$ is obtained by removing all classical edges and vertices not touching quantum edges (referred to as classical modes), leaving the disconnected quantum clusters $G_q^i$.}
\label{fig:quantumclusters}
\end{figure}

If the quantum cluster contains only $\Pi^{1,q}_e$ clauses, we can call the cluster particle-number-conserving, and states with a fixed cluster particle number span the set of possible satisfying assignments in the cluster. Some clusters are not directly particle-number conserving, but can be brought to this form by a unitary particle-hole transformation defined in Section \ref{sec:ph-trafo}. We call these clusters {\em hidden particle number conserving} (hPNC), see also Fig.~\ref{fig:Tqi_graph}.

\begin{definition}[Hidden Particle Number Conserving Cluster (hPNC)]
A quantum cluster $G_q$ on the set of vertices $V_q$ is hidden particle number conserving (hPNC) if we can divide the set of vertices into two sets, $V_q=A \cup B$, and all quantum clauses connecting vertices from $A$ with vertices from $B$ are $\Pi^{02,q}_e$ clauses, while all quantum clauses connecting vertices in either $A$ or $B$ internally, are $\Pi^{1,q}_e$ clauses. A special case is when there are no $\Pi^{02,q}_e$ clauses whatsoever in the quantum cluster, ---$A$ is the set of all vertices, and $B=\emptyset$---, and $G_q$ is particle number conserving (PNC). 
\label{def:hiddenPNCcluster}
\end{definition}

Observe that in a hPNC cluster $G_q$, one can assume wlog that all edges are single clauses (rank-1): if they are of rank-2 or higher, than either, they are classical edges (hence they had been removed from the graph $G$, hence cluster $G_q$), or they are a sum of a $\Pi^{02,q}_e$ clause and a $\Pi^{1,q}_e$ clause, and hence the cluster cannot be hPNC. Note that if we have the cluster parity operator $\hat{P}_q$ for $n_q$ modes, then the particle-hole transformation in Eq.~\eqref{eq:ph_transformation} for a subset $S \subseteq [n_q]$, transforms $K_S \hat{P}_q K_S^{-1}=\pm \hat{P}_q$: hence an eigenstate of cluster parity is still an eigenstate of cluster parity after a particle-hole transformation, but possibly with a flipped eigenvalue.

In order to efficiently determine whether $G_q$ is hPNC and examine some other properties later, it will be useful to define a ``maximal spanning" subgraph $T_q$ which is hPNC as follows. The idea is that this (non-unique) subgraph $T_q$ can be obtained via a greedy approach and can be used to efficiently decide whether a quantum cluster $G_q$ is hPNC. 

\begin{definition}[Maximal spanning hPNC subgraph of a quantum cluster]
Given a quantum cluster $G_q$ (with $n_q$ vertices and $|E_q|$ quantum clauses/edges), the graph $T_q$ is a graph with the same vertices as $G_q$ but a subset of quantum clauses of $G_q$ such that (1) $T_q$ is hPNC, (2) The quantum clauses in $T_q$ generate a spanning tree for the graph $G_q$ (i.e. all vertices in $G_q$ are connected by some quantum clauses in $T_q$) and (3) $T_q$ is maximal, i.e. it becomes non-hPNC when adding any more clauses from $G_q$.
Such (generally non-unique) $T_{q}$ can be constructed with effort $O(n_q+|E_{q}|)$ as follows. 
\vspace{0.1cm}

\noindent
$\to$ Start with $T_{q}$ containing all $n_q$ vertices, but no clauses. Pick a vertex and assign it to $A$. 

\noindent
$\to$ For each vertex adjacent to the current one, add it to $A$ if the connecting clause is of type $\Pi^{1,q}_{e}$, and add it to $B$ if it is of type $\Pi^{02,q}_{e}$ (if it is the sum of $\Pi^{1,q}_{e}$ and $\Pi^{02,q}_{e}$, pick one of them at random). Add these connecting clauses to $T_{q}$ and move to the neighbor vertex.

\noindent
$\to$ Repeat this procedure for all vertices adjacent to the previous vertices (if they have not yet been assigned to $A$ or $B$) until all vertices are in the bipartition $A \cup B$. 

\noindent
$\to$ Then, for all clauses in $G_{q}$ which are not yet in $T_{q}$, do the following. Add the clause to $T_{q}$ if it is a $\Pi^{02,q}_{e}$ clause and the vertices at its ends are in opposite sets (i.e., one in $A$ and the other in $B$). Conversely, include the clause in $T_{q}$ if it is a $\Pi^{1,q}_{e}$ clause and the vertices at its ends have equal assignments.

\vspace{0.1cm}
\noindent
If after this efficient procedure, no clauses are left in $G_q$, then obviously $T_q=G_q$ and $G_q$ is hPNC.
\label{def:Tgraph}
\end{definition}

\begin{figure}[t]
\centering
\includegraphics[width=0.65\textwidth]{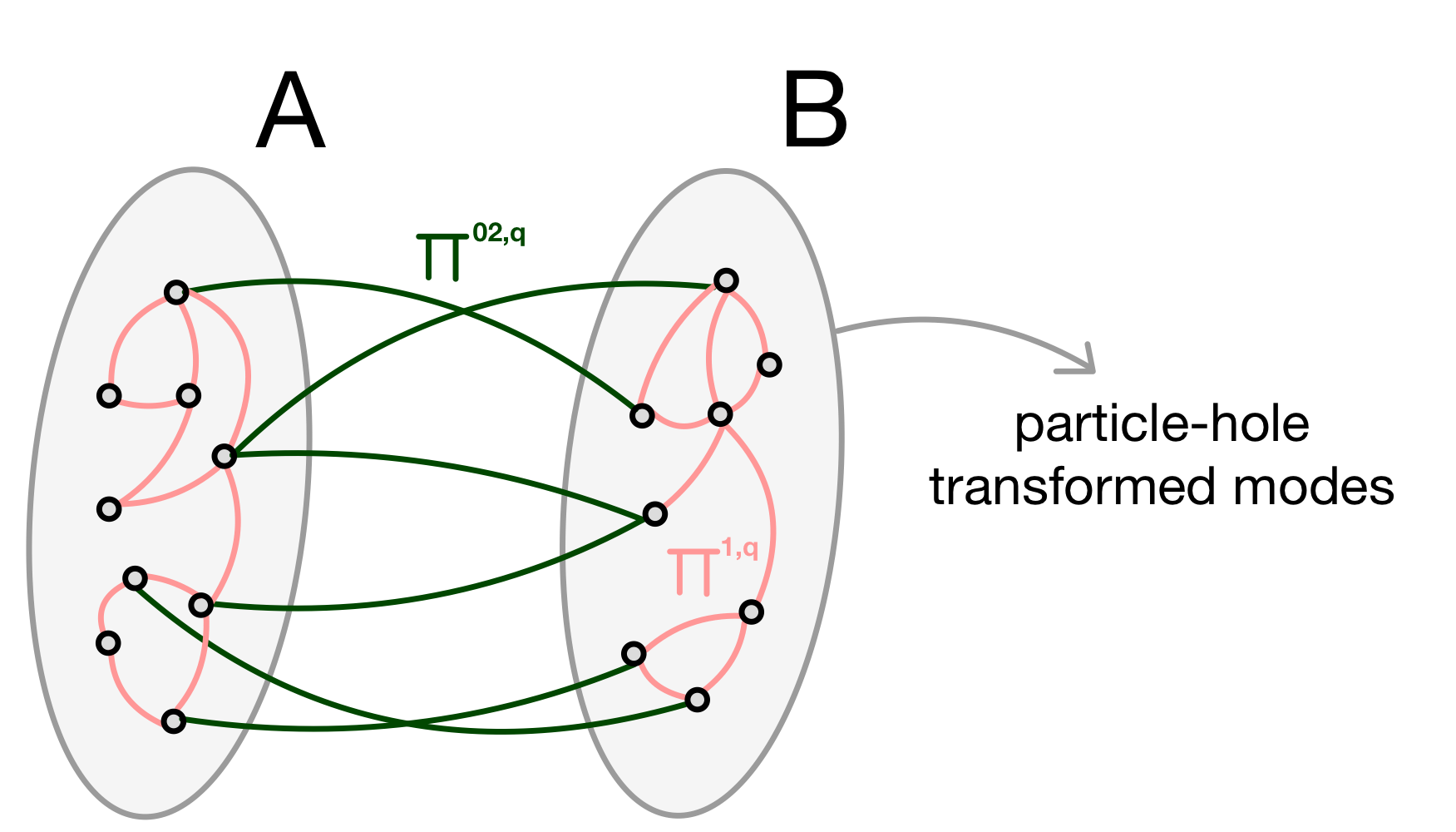}
\caption{Example of a hidden particle number conserving quantum cluster $G_{q}$ as given in Definition \ref{def:hiddenPNCcluster} where the green quantum edges are $\Pi_e^{02,q}$ clauses and the pink quantum edges are $\Pi_e^{1,q}$ clauses. If one does a particle-hole transformation $K_B$ on all the fermionic modes in $B$ (or $K_A$ on $A$), the green $\Pi_e^{02,q}$ clauses become $\tilde{\Pi}_e^{1,q}$ clauses, while the pink $\Pi_e^{1,q}$ clauses become $\tilde{\Pi}_e^{1,q}$ clauses, see Eq.~\eqref{eq:ph-edge}, thus staying of particle-conserving type. Hence, after the particle-hole transformation, all clauses in $G_q$ are particle-number-conserving, and the space of satisfying assignments is then spanned by eigenstates $\ket{\psi}$ of the so-called {\em hidden cluster particle number} $\hat{N}_q$, obeying $ \hat{N}_q (K_B \ket{\psi})=N_q (K_B \ket{\psi})$. Note that for a hPNC cluster with $B\neq \emptyset$, a satisfying assignment $\ket{\psi}$ will generally not be an eigenstate of cluster particle number in the original basis, but an eigenstate of the cluster parity.}
\label{fig:Tqi_graph}
\end{figure}

For later reference, we state a simple fact about a choice for picking the starting vertex in $T_q$:
\begin{fact}
 Suppose a quantum cluster $G_{q}$ contains a vertex with degree at least $3$, then a maximal spanning hPNC subgraph $T_q$ can be efficiently constructed such that it also contains a vertex of degree at least $3$, simply by choosing the initial vertex in the greedy approach in Definition \ref{def:Tgraph} to be the vertex in $G_{q}$ of degree at least $3$. 
\label{fact:degree3_Tgraph}
\end{fact}

\section{Characterizing the satisfying assignments of Fermionic 2-SAT}
\label{sec:eff}

\subsection{Properties of satisfying assignments}

The following `Partner' Lemma captures the idea of {\em propagation of a constraint} by genuinely quantum clauses: recall that in the qubit case, a rank-1 {\em entangled} projector is always `propagating', see Definition 1 and Lemma 5 in \cite{ASSZ:linearSAT}: the assignment of a product state for one qubit uniquely fixes the assignment of a product state for the other qubit, given that the state has to be orthogonal to the entangled projector. Here propagation is a more general concept as we can't assume a product state assignment. The Lemma roughly means that if a satisfying assignment for the whole graph $G$ has support on a certain $n$-mode state, it also needs to have support on a `partnered' $n$-mode state.

\begin{lemma}[Partner Lemma]
Given an edge $e = (j,k)$ (with $j<k$), with (1) clause $\Pi^{1,q}_{e}$ or (2) $\Pi^{02,q}_{e}$ (or both) in an instance of \textsc{Fermionic} $2$-SAT. Suppose that $\ket{\psi}$ is an $n$-mode satisfying assignment. 
\begin{enumerate}
    \item \label{1-edge} $\Pi_e^{1,q}$ clause $\Rightarrow$ If $\ket{\psi}$ has support on a $\ket{S}=a_S^{\dagger}\ket{\rm vac}$ which is given by occupations $\{x_i\}_{i=1}^n$ such that $x_{j} \neq x_{k}$, then $\ket{\psi}$ must also have support on a `partnered' state $\ket{S'}$ with occupations $\{x_i'\}_{i=1}^n$ with $x'_j=\overline{x_j},x'_k=\overline{x_k},x_{i\neq j,k}'=x_i$.
     In particular, we must have for the amplitudes $\alpha_S$ in the satisfying assignment:
    \begin{align}
    \alpha_{S} = - \text{\rm sign}(k,S')/\text{\rm sign}(j,S)\: u_{e}\:\alpha_{S'},   \mbox{ if } x_{j} = 1 \mbox{ and } x_{k} = 0, \notag \\
    \alpha_{S'} = - \text{\rm sign}(k,S)/\text{\rm sign}(j,S')\: u_{e}\:\alpha_{S}, \mbox{ if } x_{j} = 0 \mbox{ and } x_{k} = 1, 
    \label{eq:coefficientrelationquantum1}
    \end{align} 
    where $\text{sign}(l,T)$ for $l\in T$ denotes the sign arising from the reordering of $a_{T}^{\dagger}$ to $a_{l}^{\dagger}a_{T\backslash l}^{\dagger}$ and $u_e$ as in Eq.~\eqref{eq:defue} for $\Pi^{1,q}_e$.
    \item \label{02-edge} $\Pi_e^{02,q}$ clause $\Rightarrow$ If $\ket{\psi}$ has support $\ket{S}$ with occupations $\{x_i\}_{i=1}^n$ such that $x_{j} = x_{k}$, then $\ket{\psi}$ must also have support on $\ket{S'}$ with $\{x_i'\}_{i=1}^n$ with $x'_j=\overline{x_j},x'_k=\overline{x_k},x_{i\neq j,k}'=x_i$. In particular, we must have for the amplitudes $\alpha_S$ in the satisfying assignment:
    \begin{align}
    \alpha_{S} = - \text{\rm sign}(jk,S)\: v_{e}\:\alpha_{S'},   \mbox{ if } x_{j} = x_{k} = 1, \notag \\
    \alpha_{S'} = - \text{\rm sign}(jk,S')\: v_{e}\:\alpha_{S}, \mbox{ if } x_{j} = x_{k} = 0,
    \label{eq:coefficientrelationquantum02}
    \end{align}
    where ${\rm sign}(lm,T)$ for $l,m\in T$ denotes the sign arising from reordering of $a_{T}^{\dagger}$ to $a_{l}^{\dagger}a_{m}^{\dagger}a_{T\backslash \{l,m\}}^{\dagger}$ and $v_e$ as in Eq.~\eqref{eq:defue} for $\Pi^{02,q}_e$.
\end{enumerate}
\label{lemma:partners}
\end{lemma}

\begin{proof}
Consider case \ref{1-edge}. Since $\ket{\psi} = \sum_{T \subseteq [n]}\alpha_{T}a_{T}^{\dagger}\ket{\rm vac}$ is a satisfying assignment, we must have $\bra{\psi}\Pi^{1,q}_{e}\ket{\psi} = 0$. Using Eq.~\eqref{eq:projectordef_1} for $\Pi^{1,q}_{e=(j,k)}$ ($j <k$), this implies
\begin{equation}
    \bra{\psi}\Pi^{1,q}_{e}\ket{\psi} = \hspace{-0.4cm} \sum_{T\subseteq [n] \text{ s.t. }j\in T,\: k\notin T} \hspace{-0.4cm} \bigl\lvert \: \text{sign}(j,T)\: \beta_{e}^{*}\alpha_{T} + \text{sign}(k,T')\: \gamma_{e}^{*}\alpha_{T'} \bigr\rvert^{2} = 0, 
    \label{eq:cohcancel1}
\end{equation}
where for each subset $T = (x_{1},\ldots,x_{j} = 1,\ldots,x_{k} = 0,\ldots,x_{n})$, the partnered subset $T'$ is defined as $(x_{1},\ldots,x_{j} = 0,$ $\ldots,x_{k} = 1,\ldots,x_{n})$. The relation between $\alpha_S$ and $\alpha_{S'}$ follows from Eq.~\eqref{eq:cohcancel1}.
Consider case \ref{02-edge}. To ensure $\bra{\psi}\Pi^{02,q}_{e}\ket{\psi} = 0$ with $\Pi^{02,q}_e$ in Eq.~\eqref{eq:projectordef_02}, we must have that 
\begin{equation}
    \bra{\psi}\Pi^{02,q}_{e}\ket{\psi} = \hspace{-0.4cm} \sum_{T\subset [n] \text{ s.t. }j,k\in T} \hspace{-0.4cm} \bigl\lvert \: \alpha_{e}^{*}\alpha_{T'} + \text{sign}(jk,T)\: \delta_{e}^{*}\alpha_{T} \bigr\rvert^{2} = 0, 
    \label{eq:cohcancel02}
\end{equation}
where for each subset $T = (x_{1},\ldots,x_{j} = 1,\ldots,x_{k} = 1,\ldots,x_{n})$, the partnered subset $T'$ is defined as $(x_{1},\ldots,x_{j} = 0,$ $\ldots,x_{k} = 0,\ldots,x_{n})$. The relation between $\alpha_S$ and $\alpha_{S'}$ follows from Eq.~\eqref{eq:cohcancel02}.

\end{proof}

Lemma \ref{lemma:partners} strongly restricts the possible satisfying assignments on quantum clusters. The idea is that for a hPNC cluster $G_q$ repeatedly applying the case \ref{1-edge} partner rule in Lemma \ref{lemma:partners} makes particles propagate on the connected quantum cluster graph $G_q$, so that all occupations for a fixed hidden cluster particle number are generated in the cluster. Then given a fixed hidden cluster particle number, there is at most 1 unique satisfying assignment. 

If the cluster is non-hPNC, one can argue that there will be some additional $\Pi^{02,q}_e$ clauses, creating and annihilating pairs of particles via case \ref{02-edge}, and in that case there is at most 1 unique satisfying assignment given fixed parity (and no satisfying assignments with fixed hidden cluster particle number).

In these arguments, care must be taken not to assume anything about the structure of global assignment, i.e. we derive the necessary structure of such global assignment based on how quantum clauses on each cluster can be satisfied. We capture these ideas in the following Corollary:

\begin{corollary}
Suppose a quantum cluster $G_{q}$ is hPNC with bipartition $A \cup B$, and $\ket{\psi}$ is a satisfying assignment for the entire graph $G$ containing $G_q$. If $K_B\ket{\psi}$ (with $K_B=I$ in the strictly PNC case) has support on a $n$-mode state $\ket{S}$ with hidden cluster particle number $N_q$ on $G_{q}$, i.e. $\hat{N}_q \ket{S}=N_q \ket{S}$, then it must have support on all $\binom{n_q}{N_q}$  $n$-mode states $\ket{S'}$ with particle number $N_q$ on $G_{q}$ where $\ket{S'}$ has equal occupation as $\ket{S}$ on all other modes in $G$. Moreover, there is at most one satisfying assignment per hidden cluster particle number $N_q\in \{0,1,\ldots,n_q\}$ on $G_{q}$ alone.
Thus for a hPNC cluster, $N_q$ is the cluster particle number of the particle-hole transformed cluster, which we will refer to as the hidden cluster particle number. \\
\\
Suppose a quantum cluster $G_{q}$ is non-hPNC and $\ket{\psi}$ is a satisfying assignment for $G$. If $\ket{\psi}$ has support on a $n$-mode state $\ket{S}$ with cluster parity $P_q$ on $G_{q}$, i.e. $\hat{P}_q \ket{S}=P_q \ket{S}$, then it must have support on all $2^{n_q}/2$ $n$-mode states $\ket{S'}$ with parity $P_q$ on $G_{q}$ where $\ket{S'}$ has equal occupation as $\ket{S}$ on all other modes in $G$. There is at most one satisfying assignment per parity $P\in \{-1,+1\}$ on $G_{q}$ alone.
\label{cor:allparticlenumberallparity}
\end{corollary}
\begin{proof}
The first part of the Corollary follows directly from repeatedly applying `the propagation rule' of case 1 of Lemma \ref{lemma:partners} to an $N_q$-particle $n$-mode state $\ket{S}$ on which $K_B\ket{\psi}$ is supported, propagating the particles (only for $N_q=0$ (all empty) and $N_q=n_q$ there is nothing to propagate and case 1 does not apply). The coefficient relation Eq.~\eqref{eq:coefficientrelationquantum1} implies that there is at most one satisfying assignment per hidden cluster particle number $N_q\in \{0,1,\ldots,n_q\}$ on the vertices of $G_{q}$ alone.
As for the second part of the corollary, consider a non-hPNC cluster and construct the maximal spanning hPNC subgraph $T_q$ in Definition \ref{def:Tgraph} with bipartition $A \cup B$. If we were to apply the particle-hole transformation $K_B$ on $T_q$, any satisfying assignment $\ket{\psi}$ for the projectors in this subgraph $T_q$ would be such that $K_B \ket{\psi}$ is a unique state for fixed cluster particle number as per the first part of this corollary. Since $G_q$ is non-hPNC, using this bipartition of $T_q$, there will be (a) at least one internal $\Pi_e^{02,q}$ clause inside either $A$ or $B$, or (b) one $\Pi_e^{1,q}$ clause connecting $A$ and $B$. If we apply the particle-hole transformation $K_B$, then in case (a) such $\Pi_e^{02,q}$ clause stays a $\tilde{\Pi}_e^{02,q}$ ``pair-creating" clause, implying that any satisfying assignment $\ket{\psi}$ for the quantum clauses in $G_q$ has the property that $K_B\ket{\psi}$ is a unique state given some cluster parity: this follows from case 2 in Lemma \ref{lemma:partners} and states of any particle number with that cluster parity must occur in the superposition (working in the particle-hole transformed basis). If $K_B\ket{\psi}$ is a state which involves all particle numbers with given parity, then so is the satisfying assignment $\ket{\psi}$, since $K_B$ maps any state $\ket{S}$ onto some other state $\ket{S'}$.
In case (b) upon the particle-hole transformation $K_B$, the connecting clause $\Pi_e^{1,q}$ transforms to ``pair-creating" clause $\tilde{\Pi}_e^{02,q}$ and then the same arguments apply as in case (a).
\end{proof}

\subsection{Cluster-product form of satisfying assignments}
\label{sec:clus-product}

Let's examine the consequences of Corollary \ref{cor:allparticlenumberallparity} for satisfying assignments for the entire graph $G$. Let $\mathsf{hPNC}$ be the set of hPNC quantum clusters, let $\mathsf{non-hPNC}$ be the set of non-hPNC quantum clusters and let $\mathsf{Class}$ be the collection of remaining modes not contained in any quantum cluster, to which we have referred as classical modes. Let's us fix an ordering of the fermionic modes in the graph $G$ using these sets $\mathsf{hPNC},\mathsf{non-hPNC}$ and $\mathsf{Class}$ and some chosen internal ordering of modes inside the sets. Say, we first assign labels $1,2,\ldots, n_q^1$ to the modes in the first hPNC quantum cluster $G_q^1$, then labels $n_{q}^1+1,n_{q}^1+2,\ldots, n_{q}^1+n_{q}^2$ to the modes in the second hPNC quantum cluster $G_q^2$, etc., then do the same for the non-hPNC quantum clusters, and then finally for all the classical modes. What is relevant is that in this ordering, we order the modes in each cluster immediately following each other, so the satisfying assignment of the cluster itself can be used in the assignment for the whole graph in so-called cluster-product form. Given this ordering, Corollary \ref{cor:allparticlenumberallparity} implies that {\em any} satisfying assignment $\ket{\psi}$ must be a superposition of states of ordered cluster-product form, i.e.
\begin{align}
    \ket{\psi} = \sum_{\{N_q^{i}\},\{P_q^{j}\},\{x_{k}\}} \hspace{-0.2cm} \beta_{\{N_q^{i}\},\{P_q^{j}\},\{x_k\}}\notag \\
     \left(\prod_{k\in \mathsf{Class}} O_k(x_k)\right)\left(\prod_{j\in \mathsf{non-hPNC}} O_j(P_q^j)\right)\left(\prod_{i\in \mathsf{hPNC}} O_i(N_q^i)\right) 
    \ket{\rm vac},
    \label{eq:superposproduct}
\end{align}  
Here the operator $O_i(N_q^i)$ creates the unique state with hidden cluster particle number $N_q^i$ using the modes of hPNC quantum cluster $i$, the operator $O_j(P_q^j)$ creates the unique state with cluster parity \footnote{Remember that cluster parity is defined in the original basis while {\em hidden} cluster particle number is defined in the particle-holed transformed basis in Definition \ref{def:quantumclusters} and Fig.~\ref{fig:Tqi_graph}.} $P_q^j=\pm 1$ using the modes in the non-hPNC cluster $j$, and the operator $O_k(x_k=0)=I$, $O_k(x_k=1)=a_k^{\dagger}$ for a classical mode $k \in \mathsf{Class}$. Note that the operators $O_{k}(x_k)$, $O_j(P_q^j)$, $O_i(N_q^i)$ do not generally commute. We note that for a hPNC cluster $i$, non-hPNC cluster $j$ and $k\in \mathsf{Class}$, some values for respectively $N_q^{i}$, $P_q^{j}$ and $x_k$ in Eq.~\eqref{eq:superposproduct} may be excluded in order to be a satisfying assignment. For example, this can be due to further constraints on parity or particle number which we will examine in the next Section \ref{sec:ferm-qubit}. We only consider the constraints due to classical clauses in the proofs of Theorems \ref{theorem:parityconservingP} and \ref{theorem:parityconstrainedP} in Section \ref{sec:fixedparity}.

Let us now argue that taking a superposition of cluster-product states as in Eq.~(\ref{eq:superposproduct}) is unnecessary, i.e. individual cluster-product states suffice. For this, we first examine some basic properties of the cluster-product form.
For any $i_0 \in \mathsf{hPNC}$, the cluster-product form is an eigenstate of the hidden cluster particle number $\hat{N}_q^{i_0}$ (where $K_B^{i_0}$ applies the particle-hole transformation for cluster $i_0$), i.e.
\begin{align}
    \hat{N}_q^{i_0}\left(K_B^{i_0} \left(\Pi_{k\in \mathsf{Class}} O_k(x_k)\right)\left(\Pi_{j\in \mathsf{non-hPNC}} O_j(P_q^j)\right)\left(\Pi_{i\in \mathsf{hPNC}} O_i(N_q^i)\right)\ket{\rm vac}\right)= \notag \\ \left(\Pi_{k\in \mathsf{Class}} O_k(x_k)\right)\left(\Pi_{j\in \mathsf{non-hPNC}} O_j(P_q^j)\right)\left(\Pi_{i\in \mathsf{hPNC}} [\hat{N}_q^{i} K_B^{i}]^{\delta_{i,i_0}} O_{i}(N_q^{i})\right)\ket{\rm vac}=\notag \\
    N_q^{{i_0}}\left(K_B^{i_0} \left(\Pi_{k\in \mathsf{Class}} O_k(x_k)\right)\left(\Pi_{j\in \mathsf{non-hPNC}} O_j(P_q^j)\right)\left(\Pi_{i\in \mathsf{hPNC}} O_i(N_q^i)\right)\ket{\rm vac}\right),
    \end{align}
using Eq.~\eqref{eq:ph_transformation}, i.e. we can commute these operators through the ordered product.  Similarly, for any $j_0\in \mathsf{non-hPNC}$
    \begin{align}
    \hat{P}_q^{j_0}\left(\left(\Pi_{k\in \mathsf{Class}} O_k(x_k)\right)\left(\Pi_{j\in \mathsf{non-hPNC}} O_j(P_q^j)\right)\left(\Pi_{i\in \mathsf{hPNC}} O_i(N_q^i)\right)\ket{\rm vac}\right) \notag = \\
    P_q^{j_0}\left(\left(\Pi_{k\in \mathsf{Class}} O_k(x_k)\right)\left(\Pi_{j\in \mathsf{non-hPNC}} O_j(P_q^j)\right)\left(\Pi_{i\in \mathsf{hPNC}} O_i(N_q^i)\right)\ket{\rm vac}\right),
    \end{align}
    and for any mode $k_0\in \mathsf{Class}$, we have
    \begin{align}
    \hat{N}_q^{k_0}\left(\left(\Pi_{k\in \mathsf{Class}} O_k(x_k)\right)\left(\Pi_{j\in \mathsf{non-hPNC}} O_j(P_q^j)\right)\left(\Pi_{i\in \mathsf{hPNC}} O_i(N_q^i)\right)\ket{\rm vac}\right) \notag = \\
    N_q^{k_0}\left(\left(\Pi_{k\in \mathsf{Class}} O_k(x_k)\right)\left(\Pi_{j\in \mathsf{non-hPNC}} O_j(P_q^j)\right)\left(\Pi_{i\in \mathsf{hPNC}} O_i(N_q^i)\right)\ket{\rm vac}\right),
    \end{align}
   
    
    We will use the following shorthand for the cluster-product form:
\begin{equation}
    \bigl\lvert \phi\big( \{N_q^{i}\},\{P_q^{j}\},\{x_{k}\} \big) \bigr\rangle \equiv 
     \left(\prod_{k\in \mathsf{Class}} O_k(x_k)\right)\left(\prod_{j\in \mathsf{non-hPNC}} O_j(P_q^j)\right)\left(\prod_{i\in \mathsf{hPNC}} O_i(N_q^i)\right) 
    \ket{\rm vac}.
    \label{eq:sh}
\end{equation}


\begin{figure}[t]
\centering
\includegraphics[width=0.8\textwidth]{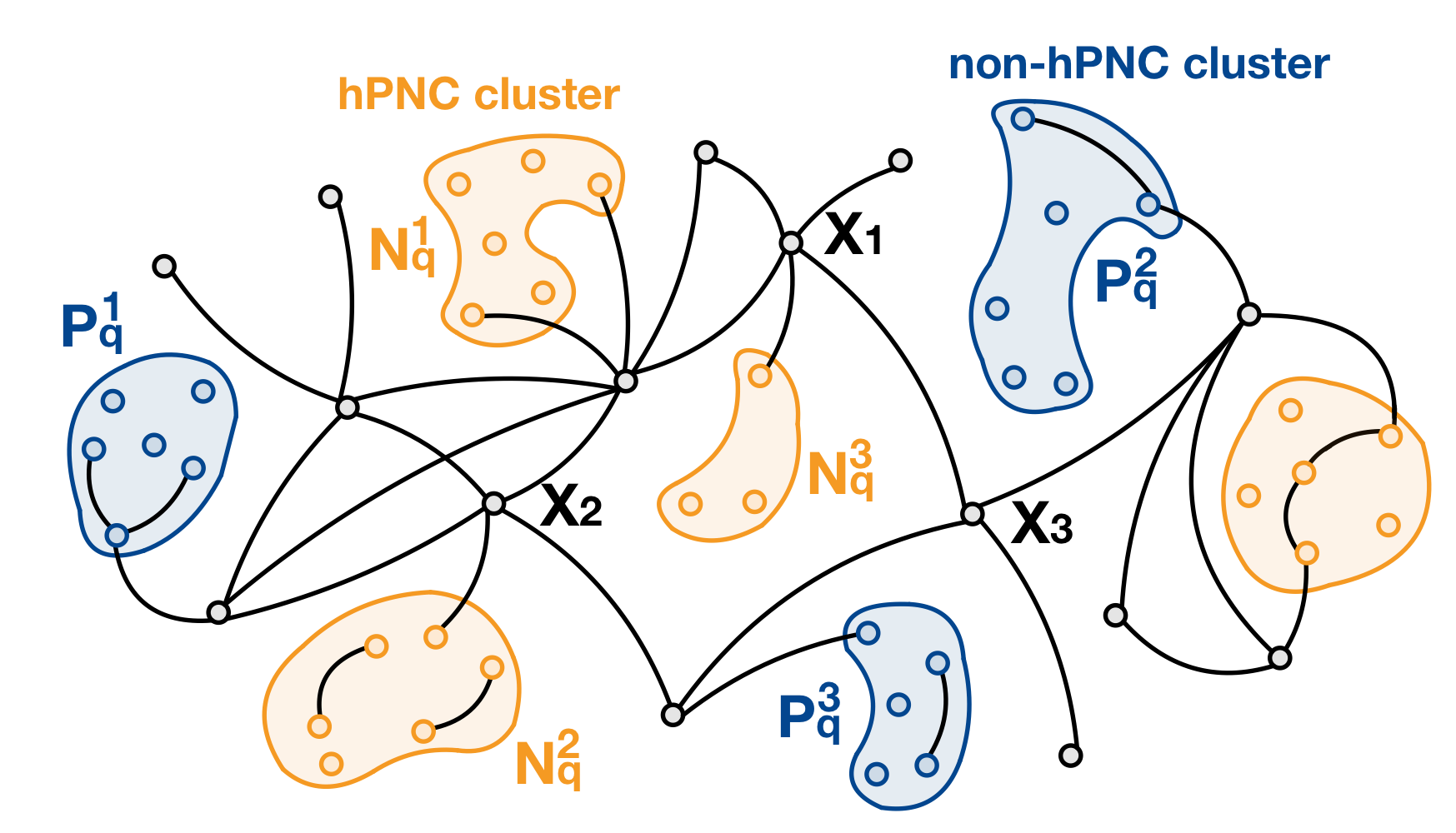}
\caption{Satisfying assignments of \textsc{Fermionic} 2-SAT can be taken to be of cluster-product form: ordered products of operators which create satisfying assignments on quantum clusters and classical assignments on classical modes starting from the vacuum state, as in Eq.~\eqref{eq:sh}. The possibly-created states on a hPNC cluster $G_q^i \equiv i \in \mathsf{hPNC}$ (orange clusters) are completely specified by (hidden) cluster particle numbers $N_q^{i}=0,1,\ldots, n_q^i$. The possibly-created states on a non-hPNC cluster $j \in \mathsf{non-hPNC}$ (blue clusters) are completely specified by the cluster parity $P_q^{j}=\pm 1$. The possibly-created states for a classical mode $k \in \mathsf{Class}$ are completely specified by the occupied/unoccupied label $x_k=0,1$. The black edges correspond to classical clauses. Note that these classical clauses can also be internal to a quantum cluster or straddle a quantum cluster and a classical mode.}
\label{fig:productform}
\end{figure}

These facts give us a final characterization of the satisfying assignments, see Fig.~\ref{fig:productform}:

\begin{proposition}[Cluster-product form of satisfying assignments]
If there exists a satisfying assignment of a \textsc{Fermionic} $2$-SAT problem, possibly with fixed particle number or fixed parity, then one can assume, with loss of generality, that it is of cluster-product form 
$\bigl\lvert \phi\big( \{N_q^{i}\},\{P_q^{j}\},\{x_{k}\} \big) \bigr\rangle$ in Eq.~\eqref{eq:sh}.
\label{lemma:global}
\end{proposition}
\begin{proof}
Given Corollary \ref{cor:allparticlenumberallparity} the most general solution is of the form in Eq.~\eqref{eq:superposproduct}. However, if any superposition state with different cluster parities $\{P_q^j\}$ for non-hPNC clusters is a satisfying assignment, then so is each individual term in the superposition since cross terms with distinct cluster parities must be zero as $\sum_i \Pi_i$ commutes with any cluster parity. Similarly, if any superposition state with different hidden particle numbers $\{N_q^i\}$ for hPNC clusters is a satisfying assignment, then so is each individual term in the superposition since (using manipulations in the particle-hole transformed basis) cross-terms must be 0. Similarly, if any superposition state with different values $\{x_k\}$ is a satisfying assignment, then so is each individual term in the superposition since cross-terms must be 0 as $\sum_i \Pi_i$ commutes with $\hat{N}_k$. Thus wlog we can restrict ourselveves to assignments of cluster-product form $\bigl\lvert \phi\big( \{N_q^{i}\},\{P_q^{j}\},\{x_{k}\} \big) \bigr\rangle$.
\end{proof}

\subsection{Excluding certain particle numbers or parities on quantum clusters}
\label{sec:ferm-qubit}

In this section we further examine the structure of satisfying assignments, namely the allowed values for (hidden) cluster particle numbers and parities in the cluster-product form in Eq.~\eqref{eq:sh} which will restrict the search for an assignment. The following Lemma provides an important simplification for quantum clusters with vertices with degree at least 3. 

\begin{lemma}[Fermionic Degree $\geq3$ Simplification]
Let $G_{q} = (V_{q},E_{q})$ be a quantum cluster on $n_q$ modes (=vertices) with at least one vertex of degree $\geq 3$ and let $T_q$ be its maximal spanning hPNC subgraph with partition $A \cup B$ in Definition \ref{def:Tgraph}. Then, for any satisfying $n_q$-mode assignment $\ket{\psi}$ for $G_{q}$, we must have $\bra{S} K_B\ket{\psi} = 0$, for any $\ket{S} = a_{S}^{\dagger}\ket{\rm vac}$ with hidden cluster particle number $\hat{N}_q=2,3,\ldots,n_q-2$. 
\label{lemma:degreeatleast3}
\end{lemma} 

\begin{proof}
Using Fact \ref{fact:degree3_Tgraph}, the graph $T_{q}$ contains a vertex $j$ of degree $\geq 3$ and let three of its neighbor vertices be labeled $k,l,m$. We may assume wlog that $j<k<l<m<i_1<i_2<\ldots$, with $i_1,i_2,\ldots$ labeling the other modes in the quantum cluster $G_q$. We write $e = (j,k)$, $e' = (j,l)$ and $e'' = (j,m)$. In the particle-hole transformed basis, all clauses adjacent to $j$ are of type $\Pi^{1,q}_e$ by construction via Definition \ref{def:Tgraph} and Fact \ref{fact:degree3_Tgraph}.
After the particle-hole transformation we write $\Pi^{1,q}_{e=(j,k)}$, $\Pi^{1,q}_{e'=(j,l)}$ and $\Pi^{1,q}_{e''=(j,m)}$ with non-zero coefficients $u_{e}$, $u_{e'}$ and $u_{e''}$ introduced in Eq.~\eqref{eq:defue}, uniquely characterizing the clauses. 

Suppose that $K_B\ket{\psi}$ has support on a state $\ket{S} = a_{S}^{\dagger}\ket{\rm vac}$ with particle number $N_q=2,3,\ldots,n_q-2$. Lemma \ref{lemma:partners} implies that $K_B\ket{\psi}$ must have support on \textit{all} $\binom{4}{2}$ states $\ket{S'} = a_{S'}^{\dagger}\ket{\rm vac}$ s.t. in $\ket{S'}$ two modes from $\{j,k,l,m\}$ are occupied, and $\ket{S'}$ has occupation equal to $\ket{S}$ on all modes $V_q\backslash \{j,k,l,m\}$. Let the occupied modes in $V_q \backslash \{j,k,l,m\}$ be denoted by subset $T$ (which depends on $S$). Let $K_B\ket{\psi}=\sum_S \alpha_S \ket{S}$, using, say, the fermionic mode ordering established at the start of Section \ref{sec:clus-product}, and assume wlog that in this ordering $T=T_{\rm before} \cup T_{\rm after}$ where $T_{\rm before}$ is a set of modes which come before $j,k,l,m$ in the chosen fermionic order and $T_{\rm after}$ is a set of modes which come after and let $a_T^{\dagger}=a_{T_{\rm before}}^{\dagger}a_{T_{\rm after}}^{\dagger}$.
Then $K_B\ket{\psi}$ must have support on each of the six states 
\begin{multline}
a_{j}^{\dagger}a_{k}^{\dagger}a_{T}^{\dagger}\ket{\rm vac},  a_{j}^{\dagger}a_{l}^{\dagger}a_{T}^{\dagger}\ket{\rm vac}, a_{j}^{\dagger}a_{m}^{\dagger}a_{T}^{\dagger}\ket{\rm vac},    a_{k}^{\dagger}a_{l}^{\dagger}a_{T}^{\dagger}\ket{\rm vac},  a_{k}^{\dagger}a_{m}^{\dagger}a_{T}^{\dagger}\ket{\rm vac}, 
a_{l}^{\dagger}a_{m}^{\dagger}a_{T}^{\dagger}\ket{\rm vac},
\label{eq:starstates}
\end{multline}
with amplitudes $\alpha_{j,k,T},\alpha_{j,l,T}$,$\alpha_{j,m,T},\alpha_{k,l,T},\alpha_{k,m,T}$ and $\alpha_{l,m,T}$. Lemma \ref{lemma:partners} implies the following relations between these amplitudes:
\begin{multline}
    \text{(1)}\: \alpha_{j,k,T} = +u_{e'}\:\alpha_{k,l,T}, \hspace{0.2cm}
    \text{(2)}\: \alpha_{j,l,T} = -u_{e}\:\alpha_{k,l,T}, \hspace{0.2cm}
    \text{(3)}\: \alpha_{j,l,T} = +u_{e''}\:\alpha_{l,m,T}, \\ 
    \text{(4)}\: \alpha_{j,k,T} = +u_{e''}\:\alpha_{k,m,T}, \hspace{0.2cm}
    \text{(5)}\: \alpha_{j,m,T} = -u_{e}\:\alpha_{k,m,T}, \hspace{0.2cm}
    \text{(6)}\: \alpha_{j,m,T} = -u_{e'}\:\alpha_{l,m,T}.
    \label{eq:ferm-con}
\end{multline}
However, conditions (1-6) imply $+1 = -1$. Hence, $K_B\ket{\psi}$ cannot have support on $n$-mode states $\ket{S} = a_{S}^{\dagger}\ket{\rm vac}$ with cluster particle number $\hat{N}_q=2,3,\ldots, n_q-2$. 
\end{proof}

We note that the previous Lemma is intrinsically fermionic, i.e. the anti-commutation of creation operators plays a role in the signs of Eq.~\eqref{eq:ferm-con} and there is no \textsc{Quantum} 2-SAT counterpart. The following corollary summarizes the consequences. 

\begin{corollary}
Consider an assignment of cluster-product form 
\begin{align}
\ket{ \phi\left( \{N_q^{i}\}_{i\in \mathsf{hPNC}},\{P_q^{j}\}_{j\in \mathsf{non-hPNC}},\{x_{k}\}_{k\in \mathsf{Class}}\right)},
\label{eq:clusprod}
\end{align}
in Eq.~\eqref{eq:sh}, which, per Proposition \ref{lemma:global}, exists if there is a satisfying assignment for $G$. We summarize some consequences for the quantum numbers $N_q^i$ and $P_q^i$ in the cluster-product form due to the quantum clauses on clusters.
\begin{enumerate}
    \item \label{hPNC3} For degree $\geq 3$ hPNC clusters $i\in \mathsf{hPNC}$, the only hidden cluster particle numbers can be $N_q^i=0,1,n_q^i-1,n_q^i$ which are Gaussian states, remaining Gaussian when undoing the particle-hole transformation as argued in Section \ref{sec:ph-trafo}.
    \item \label{alwaysemptyfull} For any hPNC cluster $i\in \mathsf{hPNC}$ with $n_q^i$ modes, the hidden cluster particle number can {\em always} be $N_q^i=0$ or $N_q^i=n_q^i$, since Lemma \ref{lemma:partners} imposes no constraints on the coefficients, i.e. particles don't propagate since there are either no particles or all modes are filled.
    \item \label{needstobefree} For a hPNC cluster $i$ with $n_q^i$ modes, let there be a cluster product assignment with hidden cluster particle number $N_q^i\neq 0, N_q^i\neq n_q^i$. Then the assignment will have support on states $\ket{S}$ in which any mode $j$ in the cluster is empty ($x_j=0$) \textit{and} states for which it is filled ($x_j=1$). Hence in this case any classical clause straddling the cluster on mode $j$ and a classical mode $k\in \mathsf{Class}$ or mode in another quantum cluster, will have to be satisfied (if it can) no matter the value for $x_j=0,1$. When $n_q^i > 2$, if the state with hidden cluster particle number $N_q^i\neq 0, N_q^i\neq n_q^i$ obeys any additional classical clause constraints inside the cluster, then either the all-empty $N_q^i=0$ or all-filled assignment $N_q^i=n_q^i$ (or both) should also obey these constraints since on any pair of modes, at least 3 out of 4 occupations on these two modes will be allowed (thus always including either 00 or 11). 
    \item If $G$ contains a degree $\geq 3$ non-hPNC cluster $j\in \mathsf{non-hPNC}$ with a number of modes $n_q^j\geq 5$, then $G$ has no satisfying assignment. Via Corollary \ref{cor:allparticlenumberallparity} the satisfying assignment is characterized by total parity $P$ and has support on states with any hidden cluster particle number with this parity. However, Lemma \ref{lemma:degreeatleast3} implies that the assignment cannot have support on odd hidden cluster particle number $N_q=3$ (in case of odd parity) or $N_q=2$ (in case of even parity) when $n_q\geq 5$, hence there cannot be any satisfying assignment.
    \item \label{degree3} For a degree $\geq 3$ non-hPNC cluster $G_q\equiv j\in \mathsf{non-hPNC}$ with $n_q^j = 4$, $K_B\ket{\psi}$ (with the set $B$ defined through the graph $T_q$ of the cluster $G_q$), can only have support on states with odd cluster parity. Undoing the particle hole transformation, this excludes one choice of parity $P$ for this cluster in Eq.~\eqref{eq:clusprod} (which one depends on whether $|B|$ is odd or even). In Appendix \ref{app:exampleNG} we show that the satisfying assignment for such cluster of four fermionic modes is a non-Gaussian fermionic state.
    \item \label{needstobefree2} For a non-hPNC $G_{q}$ cluster with $n_q$ modes, let there be an cluster product assignment with some cluster parity $P_q$. Then this assignment will have support on states $\ket{S}$ in which any mode $j$ in the cluster is empty \textit{and} states for which it is filled: it has support on all occupations $x_j$ for any $j$ in the cluster, with the given fixed parity by Corollary \ref{cor:allparticlenumberallparity}. Hence any classical clause straddling the cluster on mode $j$ and a classical mode $k\in \mathsf{Class}$ or mode in another quantum cluster, will have to be satisfied (if it can) no matter the value for $x_j=0,1$.
\end{enumerate}
\label{cor:restricted-PN}
\end{corollary}


The corollary means that if a vertex in the cluster is of degree at least $3$, then, if the cluster is non-hPNC, we must have $n_q\leq 4$, and hence it is easy to verify whether an assignment exists. Otherwise, if it is hPNC and has a degree $\geq 3$ vertex, its assignments are simple and Gaussian as they are very restricted in particle number (in the particle-hole transformed basis), collapsing the search space to being linearly dependent on the number of modes in the cluster. We capture this efficiency more precisely in Lemma \ref{lem:ccheck} in the next section.

If all vertices in the cluster $G_q$ are of degree less than 3, the graph $G_q$ is a line or a loop. For a line, one can apply the Jordan-Wigner transformation and map the problem onto quantum 2-SAT on a line which can fully be solved, using existing methods. For a loop there is a small adaptation due to the Jordan-Wigner Pauli-$Z$ strings. The cases of $G_q$ being a line or a loop are dealt with in Section \ref{sec:linesandloops} below. 

Corollary \ref{cor:restricted-PN} also makes clear that the assignments of hidden particle number and parity on the cluster lead to strong restrictions on whether and how one can satisfy the remaining classical clauses in $E_c$, and this will be used to prove our main Theorems \ref{theorem:parityconservingP}, \ref{theorem:parityconstrainedP} and \ref{theorem:NPcomplete}. 

\subsubsection{Quantum clusters: lines and loops}
\label{sec:linesandloops}
To handle lines and loops, it is convenient to switch to qubit language as \textsc{Fermionic} 2-SAT projectors on a line map directly onto \textsc{Quantum} 2-SAT projectors on qubits via the Jordan-Wigner transformation in Eq.~\eqref{eq:JW}.

Let us first consider the taxonomy of graphs $G_q$ that are lines or loops, shown in Figure \ref{fig:linelooptaxonomy}. Remember that there is at most a single $\Pi^{02,q}_e$ or single $\Pi^{1,q}_e$ per edge $e$, otherwise one can reduce the pair to a classical edge, see the arguments in Section \ref{sec:higherrank}, and we do not yet consider any classical clauses in this section yet. If $G_q$ is a line, its maximal spanning hPNC subgraph $T_q$ (see Definition \ref{def:Tgraph}) is also a line, and $G_q$ is either hPNC or non-hPNC as in Figure \ref{fig:linelooptaxonomy}(a) and (b). If $G_q$ is a loop, then $T_q$ is either a loop, see Figure \ref{fig:linelooptaxonomy}(c) and (d), or a line, Figure \ref{fig:linelooptaxonomy}(e) and (f).
\begin{figure}[H]
    \centering
    \begin{subfigure}[t]{0.16\textwidth}
        \centering
        \includegraphics[width=1.08\textwidth]{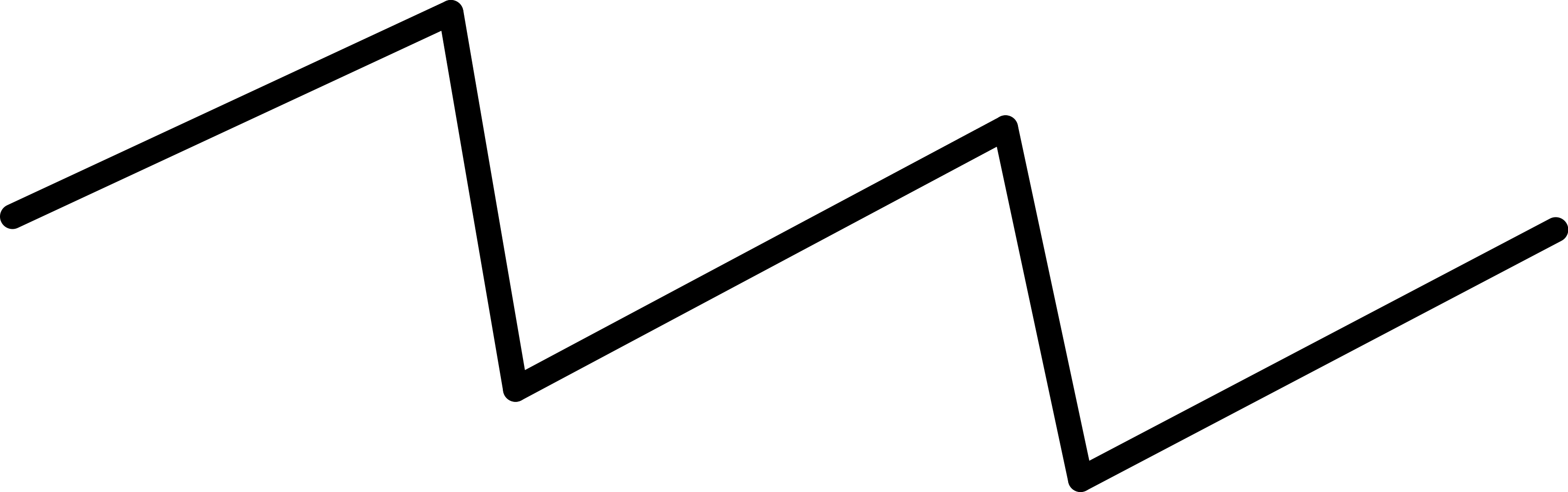}
        \caption{Lemma \ref{lem:line} (hPNC).}
    \end{subfigure}
    ~
    \centering
    \begin{subfigure}[t]{0.16\textwidth}
        \centering
        \includegraphics[width=1.08\textwidth]{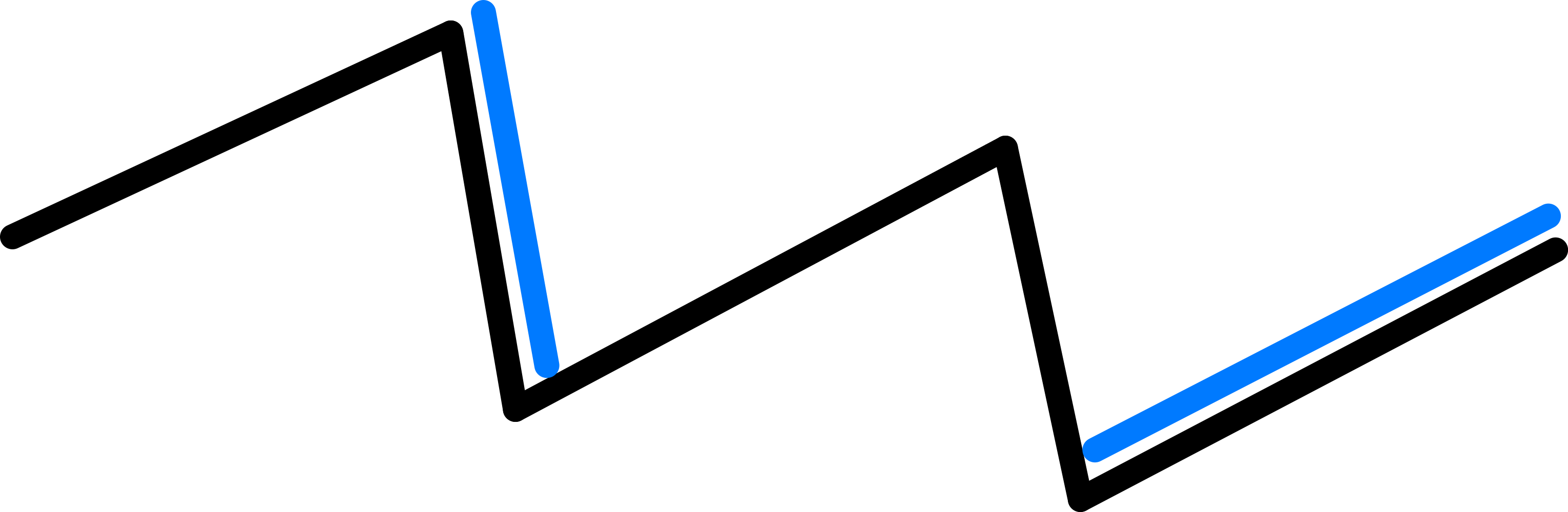}
        \caption{Lemma \ref{lem:nonhPNClinesloops} (non-hPNC).}
    \end{subfigure}
    \hspace{0.2cm}
    \begin{subfigure}[t]{0.14\textwidth}
        \centering
        \includegraphics[width=0.75\textwidth]{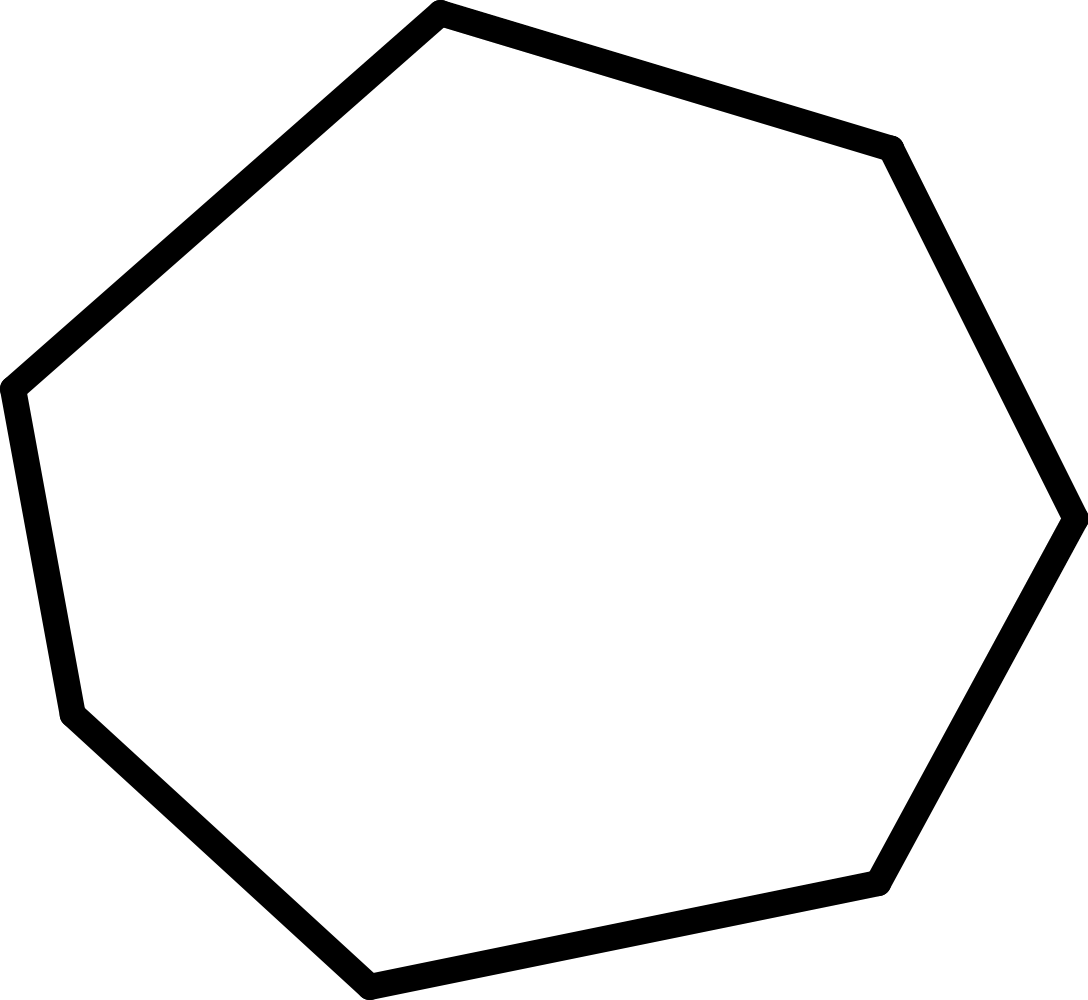}
        \caption{Lemma \ref{lem:loop} (hPNC).}
    \end{subfigure}
    ~
    \begin{subfigure}[t]{0.14\textwidth}
        \centering
        \includegraphics[width=0.75\textwidth]{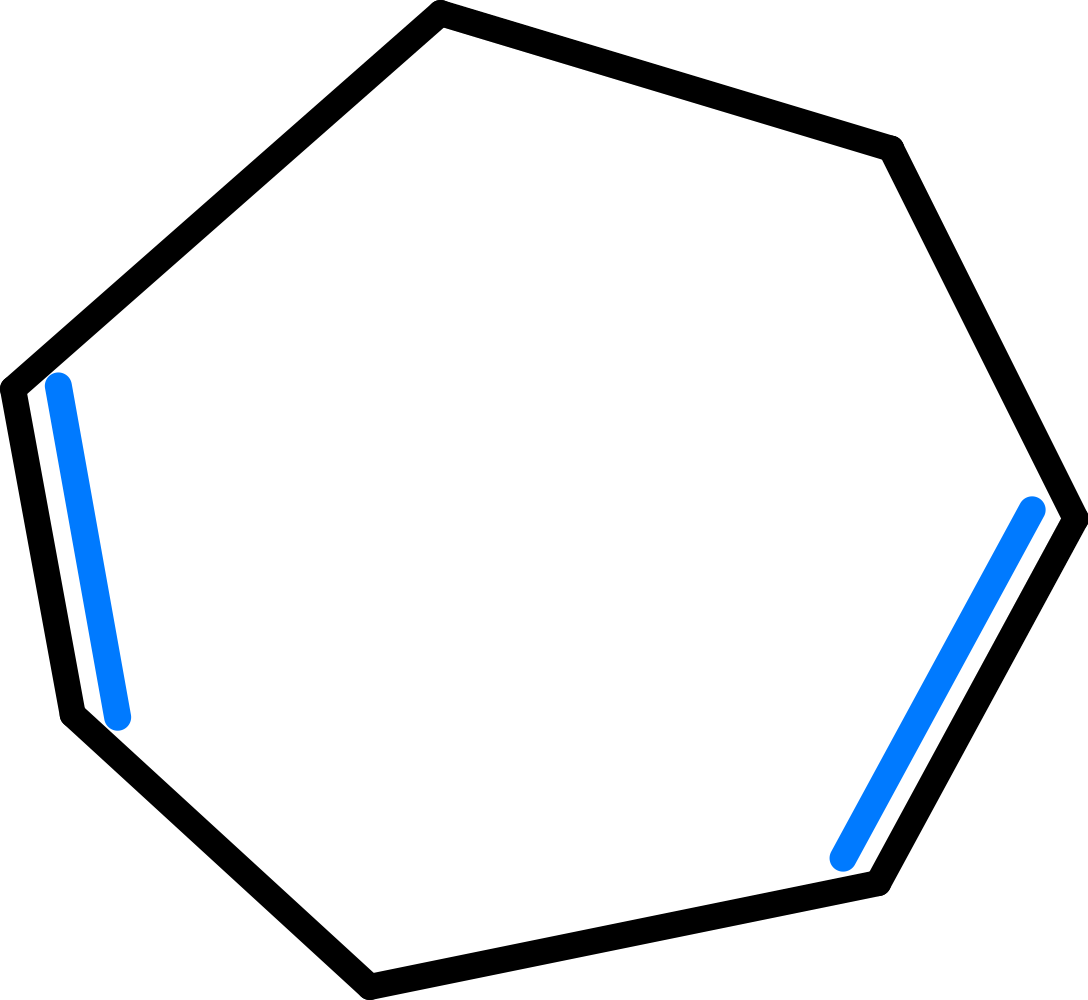}
        \caption{Lemma \ref{lem:nonhPNClinesloops} (non-hPNC).}
    \end{subfigure}%
    ~ 
    \begin{subfigure}[t]{0.14\textwidth}
        \centering
        \includegraphics[width=0.75\textwidth]{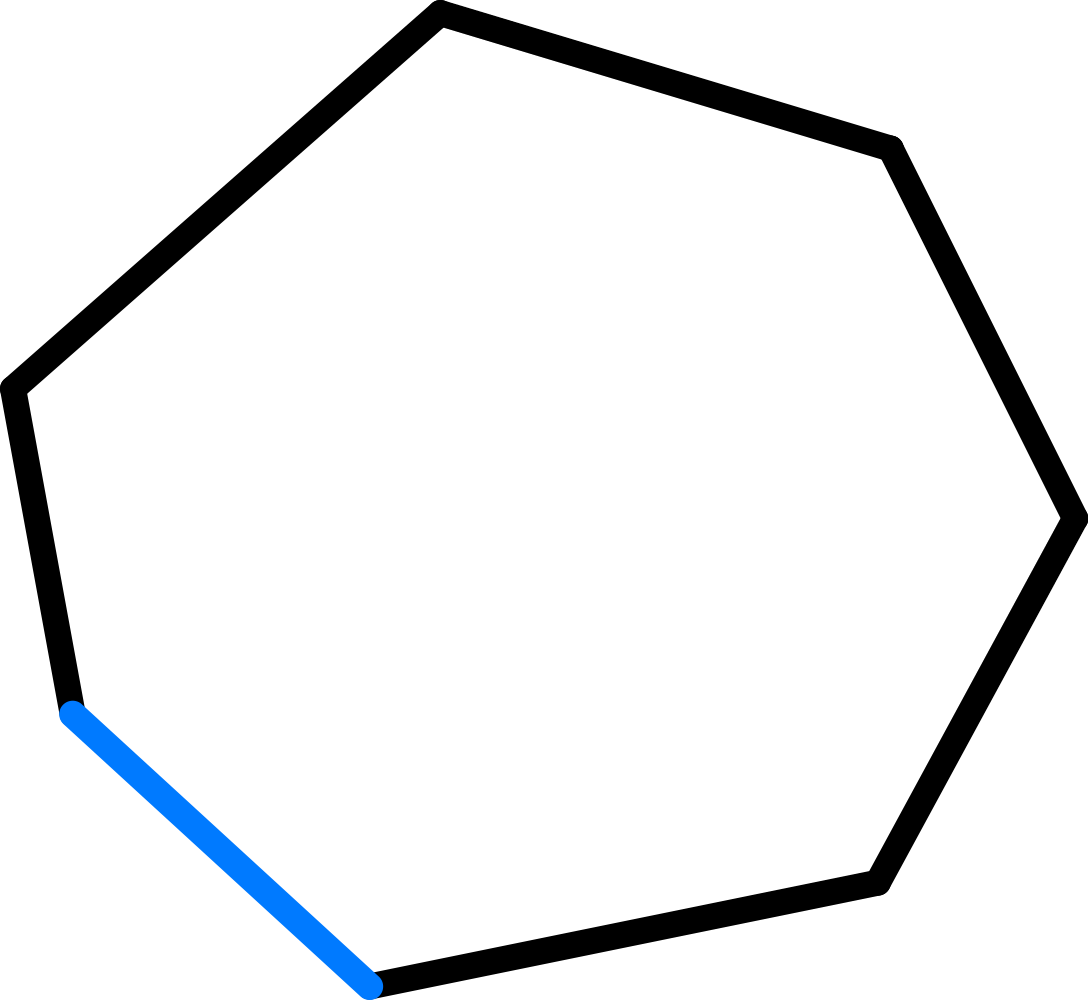}
        \caption{Lemma \ref{lem:nonhPNClinesloops} (non-hPNC).}
    \end{subfigure}
    ~
    \begin{subfigure}[t]{0.14\textwidth}
        \centering
        \includegraphics[width=0.75\textwidth]{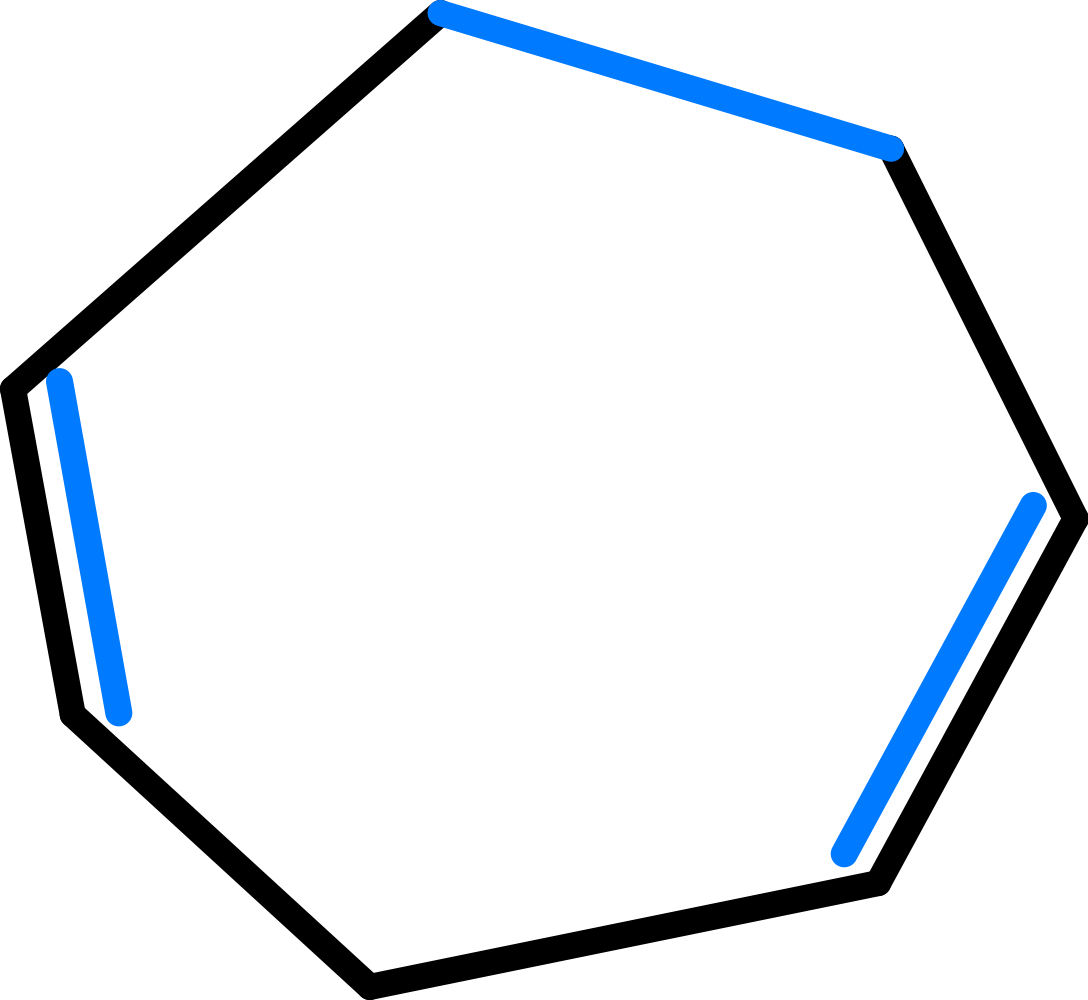}
        \caption{Lemma \ref{lem:nonhPNClinesloops} (non-hPNC).}
    \end{subfigure}
    \caption{Different types of line and loop clusters in the particle-hole-transformed basis. The black edges make up $T_q$ (representing $\Pi^{1,q}_e$-type clauses) and the blue edges are $\Pi^{02,q}_e$-type clauses.}
    \label{fig:linelooptaxonomy}
\end{figure}

First, we consider the case of $G_q$ being a line. The following Lemma is to be applied to $T_q$ in the particle-holed transformed basis, thus containing only $\Pi^{1,q}$-type clauses. It is essentially a rephrasing of some results in \cite{JWZ:SATstructure}. 

\begin{lemma}
    Given a line with $n_q$ fermionic modes and (rank-1) clauses $\Pi_{(i,i+1)}^{1,q}$ between modes $i$ and $i+1$ for $i\in \{1,2,\ldots,n_q-1\}$. For each particle number $N_q\in \{0,1,\ldots,n_q\}$, there is a unique assignment satisfying the clauses $\Pi_{(i,i+1)}^{1,q}$. 
    \label{lem:line}
\end{lemma}

\begin{proof}
    Via the Jordan-Wigner transformation in Eq.~\eqref{eq:JW}, each 2-mode rank-1 projector $\Pi_{(i,i+1)}^{1,q}$ becomes a projector onto a two-qubit quantum state $\Pi_{i,i+1}^{\rm qubit}=\ket{\psi}\bra{\psi}_{i,i+1}$ with $\ket{\psi}_{(i,i+1)}= \alpha\ket{01} + \beta\ket{10}$ (where $|\alpha|^2+|\beta|^2 = 1$). The entangled state $\ket{\psi}_{(i,i+1)}$ can generally be expressed as $B_i \otimes I \ket{\Psi^-}$ with a two-qubit singlet state $\ket{\Psi^-}=\frac{1}{\sqrt{2}}(\ket{01}-\ket{10})$ and some diagonal, invertible $B_{i}$. 
    
    Let $H=I \otimes A_2 \ldots \otimes A_{n_q} H_{\rm FM-Heis} I \otimes A_2^{\dagger} \ldots \otimes A_{n_q}^{\dagger}\geq 0$, with $H_{\rm FM-Heis}=\sum_i \ket{\Psi^-}\bra{\Psi^-}_{i,i+1}$, the 1D ferromagnetic Heisenberg model. Here $A_2$ is uniquely chosen such that $I \otimes A_2 \ket{\Psi^-}_{1,2}=B_1 \otimes I \ket{\Psi^-}_{1,2}=\ket{\psi}_{1,2}$. Given $A_2$, $A_3$ is chosen such that $A_2 \otimes A_3 \ket{\Psi^-}_{2,3}=
    B_2 \otimes I \ket{\Psi^-}_{2,3}=\ket{\psi}_{2,3}$ etc. Importantly, the satisfying assignments of $\{\Pi_{(i,i+1)}^{\rm qubit}\}$ are the zero-eigenvalue states of $H$.
    A zero-eigenvalue state with particle (or "excitation") number $N_q$ can thus be obtained by applying $I \otimes A_2^{\dagger \:-1} \otimes \ldots \otimes A_{N_q}^{\dagger \:-1}$ to any state $\ket{\psi_{\rm FM-Heis}^{N_q}}$ in the null-space of $H_{\rm FM-Heis}$, with equal particle number $N_q$ as the $A_i^{\dagger \:-1}$ matrices are diagonal. Hence, given a null-space solution with excitation number $N_q$ for $H_{\rm FM-Heis}$, one obtains a satisfying assignment for the fermionic problem on a line.
    
    Due to the $U^{\otimes n_q}$-invariance of $H_{\rm FM-Heis}$, the eigenspaces must be spanned by states of the form $\ket{\chi^{\otimes n_q}}$, hence symmetric subspaces of $n_q$ qubits \cite{harrow}. There is only one such state with a given particle number $N_q$, namely the fully permutation-symmetric state, which is in the null-space of $H_{\rm FM-Heis}$ since $\ket{\Psi^-}\bra{\Psi^-}_{i,i+1} \ket{\psi_{\rm FM-Heis}^{N_q}}=\ket{\Psi^-}\bra{\Psi^-}_{i,i+1} \pi(i,i+1)\ket{\psi_{\rm FM-Heis}^{N_q}}=-\ket{\Psi^-}\bra{\Psi^-}_{i,i+1} \ket{\psi_{\rm FM-Heis}^{N_q}}=0$ where $\pi(i,i+1)$ is the permutation between qubits $i$ and $i+1$. Note that the unique satisfying assignment for the quantum $2$-SAT clauses is thus $I \otimes A_2^{\dagger \:-1} \otimes \ldots \otimes A_{n_q}^{\dagger \:-1}\ket{\psi_{\rm FM-Heis}^{N_q}}$.
\end{proof}

Let us now consider a hPNC cluster for which both $G_q$ and $T_q$ are a loop as in Fig.~\ref{fig:linelooptaxonomy}(c). Again the following Lemma is to be applied to $T_q$ in the particle-hole-transformed basis in which it consists of $\Pi^{1,q}_e$ clauses only. For there to be any satisfying assignments on $T_q$, the clause closing the loop needs to be consistent with the other clauses. For \textsc{Fermionic} $2$-SAT clauses this consistency condition also turns out to depend on the cluster parity: 

\begin{lemma}
    Given a loop consisting of $n_q$ fermionic modes with clauses $\Pi^{1,q}_{(i,i+1 \text{\normalfont{ mod }} n_q)}$ for $i \in \{1,2,\ldots,n_q\}$, with coefficients $u_{(i,i+1 \text{\normalfont{ mod }} n)}$ defined in Eq. \eqref{eq:defue}. Besides the satisfying assignments at $N_q = 0$ (all-empty) and $N_q = n_q$ (all-filled), there are satisfying assignments at a given $N_q$ iff $(-1)^{N_q+n_q-1} \big[ 1/u_{(1,n_q)} \prod_{i = 1}^{n_q-1}u_{(i,i+1)} \big] = 1$. So, if $| 1/u_{(1,n_q)} \prod_{i = 1}^{n_q-1}u_{(i,i+1)} | = 1$, there are only satisfying assignments at all even $N_q$ (when $n_q$ is odd) or odd $N_q$ (when $n_q$ is even).
    \label{lem:loop}
\end{lemma} 

\begin{proof}
    The proof of Lemma \ref{lem:loop} is similar to the proof of Lemma \ref{lem:line}. The only difference is that in the case of a loop, the clause $\Pi^{1,q}_{(1,n_q)}$, when Jordan-Wigner transformed, becomes 
    \begin{multline}
        2\Pi_{(1,n_q)}^{\rm qubit} = I_{1}\otimes B_{n_q}\Big(\ket{1}_{1}\ket{0}_{n_q}-Z_{2}\ldots Z_{n_q-1}\ket{0}_{1}\ket{1}_{n_q}\Big)\times \\ \Big(\bra{1}_{1}\bra{0}_{n_q}-Z_{2}\ldots Z_{n_q-1}\bra{0}_{1}\bra{1}_{n_q}\Big)I_{1}\otimes B_{n_q}^{\dagger},
        \label{eq:projqubit}
    \end{multline}
    with $B_{n_q}$ diagonal and  $(B_{n_q})_{11}/(B_{n_q})_{22} = -\beta_{(1,n_q)}/\gamma_{(1,n_q)} = -1/u^{*}_{(1,n_q)}$ (see Eq. \eqref{eq:defue}).
    The satisfying assignments of the loop problem correspond to zero-energy states of $H$ in the proof of Lemma \ref{lem:line} that are also projected to zero by $\Pi_{(1,n_q)}^{\rm qubit}$. We have that $(B_{i})_{11}/(B_{i})_{22} = -\gamma_{(i,i+1)}/\beta_{(i,i+1)} = -u^{*}_{(i,i+1)}$ for $i\in \{1,2,\ldots,n_q-1\}$, from which the entries of $A_{2},\ldots,A_{n_q}$ follow (see the proof of Lemma \ref{lem:line}). 
    The relation $\Pi_{(1,n_q)}^{\rm qubit} I\otimes A^{\dagger \: -1}_{2} \otimes \ldots \otimes A^{\dagger \: -1}_{n_q}\ket{\psi^{N_q}_{\rm FM Heis}} = 0$ holds iff $u_{(1,n_q)} = (-1)^{N_q}(A_{n_q})_{22}^*/(A_{n_q})_{11}^*$ is obeyed (where we have used $(B_{n_q})_{11}/(B_{n_q})_{22} = -1/u^{*}_{(1,n_q)}$). Using that the entries of $A_{2},\ldots,A_{n_q}$ can be expressed in terms of $\{u_{(i,i+1)}\}_{i=1}^{n_q-1}$, this condition is equivalent to the condition $(-1)^{N_q+n_q-1} \big[ 1/u_{(1,n_q)} \prod_{i = 1}^{n_q-1}u_{(i,i+1)} \big] = 1$ from the lemma statement. 
\end{proof}

\subsubsection{non-hPNC lines and loops}

We have thus far considered the hPNC line or loop graphs. For non-hPNC lines and loops, however, we must consider the effect of the $\Pi^{02,q}_e$ clauses (in the particle-hole-transformed basis) on the satisfiability of these clusters as in Fig.~\ref{fig:linelooptaxonomy}(b),(d), (e) and (f).
These $\Pi^{02,q}_e$ clauses can be of two types: they can either be part of a rank-2 projector as in Fig.~\ref{fig:linelooptaxonomy}(d), or they can close a loop as in Figure \ref{fig:linelooptaxonomy}(e). We will refer to the former as a \textit{non-loop-closing} $\Pi^{02,q}_e$ clause and to the latter as a \textit{loop-closing} $\Pi^{02,q}_e$ clause. Of course, any loop $G_q$ contains at most one loop-closing $\Pi^{02,q}_e$ clause. If $G_q$ is a line, then all $\Pi^{02,q}_e$ clauses are non-loop-closing as in Fig.~\ref{fig:linelooptaxonomy}(b). If $G_q$ and its maximal spanning hPNC subgraph $T_q$ are loops, all $\Pi^{02,q}_e$ clauses are also non-loop-closing (Fig.~\ref{fig:linelooptaxonomy}(d)). 

Note that one can always apply a different particle-hole transformation (with a different $T_q$) such that the instance in Figure \ref{fig:linelooptaxonomy}(f) becomes an instance in Figure \ref{fig:linelooptaxonomy}(d), by ensuring that the loop is closed on a rank-2 edge. Therefore, if $G_q$ contains a loop-closing $\Pi^{02,q}_e$ clause, then wlog it is the only $\Pi^{02,q}_e$ clause in $G_q$. 

Lemma \ref{lem:nonhPNClinesloops} below provides a detailed account of conditions under which a parity $P_q$ is allowed on a quantum cluster $G_q$, provided that $P_q$ is allowed on $T_q$. The important take-away is that for any non-hPNC line or loop, these conditions come down to straightforward relations between the $u_e$ and $v_e$ coefficients of respectively $\Pi^{1,q}_e$ and $\Pi^{02,q}_e$ clauses along the line or loop. 

\begin{lemma}
Given a non-hPNC cluster $G_q$ which is a line or a loop on $n_q$ modes. The following statements hold in the particle-hole-transformed basis in which $T_q$ only has $\Pi^{1,q}_e$ clauses and $P_q$ refers to the parity of the particle-hole-transformed assignment. Assume that $T_q$ allows for satisfying assignments of parity $P_q$. If $G_q$ contains a single $\Pi^{02,q}_{e}$ clause, then the same parities of $T_q$ are allowed as satisfying assignments for $G_q$. If $G_q$ contains at least two $\Pi^{02,q}_{e}$ clauses, then suppose:
\begin{enumerate}
\item $G_q$ is a line (Figure \ref{fig:linelooptaxonomy}(b)) then any $P_{q}$ of $T_q$ is allowed for $G_q$ iff the following consistency condition is satisfied. For each pair of subsequent non-loop-closing $\Pi^{02,q}_e$ clauses on edges $(j,j+1)$ and $(k,k+1)$ ($j<k$), 
\begin{equation}
(\star)\;\;\; v_{(j,j+1)} = \Big[ \prod_{i=j}^{k-1}u_{(i,i+1)}u_{(i+1,i+2)}\Big]v_{(k,k+1)} 
\label{eq:nonloopclosingcondition}
\end{equation} 
must hold. 
\item $G_q$ is a loop and $T_q$ a line (Figure \ref{fig:linelooptaxonomy}(e). $P_q$ is always allowed on $G_q$ since there is only one $\Pi^{02,q}_e$ clause in $G_q$. 
\item Both $G_q$ and $T_q$ are loops (Figure \ref{fig:linelooptaxonomy}(d)). Any $P_{q}$ of $T_q$ is allowed for $G_q$, if condition (\:$\star$) holds on $\Pi^{02,q}_e$ clauses along the loop, with the following caveat. If one of the non-loop-closing clauses is on edge $(1,n_q)$, then the one condition (on subsequent pairs) involving $\Pi^{02,q}_{(1,n_q)}$ has $v_{(1,n_q)}$ replaced by $-1/u_{(1,n_q)}^{2}v_{(1,n_q)}$ on the RHS of (\:$\star$). 
\end{enumerate}
\label{lem:nonhPNClinesloops}
\end{lemma}
\begin{proof}
By assumption, $P_q$ is allowed on $T_q$. According to Corollary \ref{cor:allparticlenumberallparity}, if $P_q$ is allowed on $G_q$, then a satisfying assignment has support on all states $\ket{S}=a_{S}^{\dagger}\ket{\rm vac}$ with parity $P_q$. For $P_{q}$ to be allowed on $G_{q}$, the $\Pi^{02,q}_e$ clauses on $G_q$ should not lead to inconsistencies when connecting these states of parity $P_q$ via Lemma \ref{lemma:partners}. Such inconsistencies can come up when, starting with any $\ket{S}$, one repeatedly applies Lemma \ref{lemma:partners} to end up at state $\ket{S}$ again. If the associated coefficient relations in Lemma \ref{lemma:partners} are not consistent, $P_q$ is excluded on $G_q$. For $G_q$ lines and loops containing a single $\Pi^{02,q}_e$ clause, these inconsistencies are always avoided. For non-hPNC lines and loops with multiple $\Pi^{02,q}_e$ clauses, the consistency of these relations can be checked rather straightforwardly. To see this, we distinguish between cases (1), (2) and (3) in the lemma statement. 

For case (1), all $\Pi^{02,q}_e$ clauses are non-loop-closing. The following process, realized by repeated application of Lemma \ref{lemma:partners}, can lead to inconsistency. Start off with $\ket{S}$ with equal occupation on an edge with a $\Pi^{02,q}_e$ clause. One creates a pair of particles/holes on edge $e$, migrates them to a clause $\Pi^{02,q}_{e'}$ at edge $e'$ via $\Pi^{1,q}_e$ clauses, and undoes the creation of the particle/hole pair through clause $\Pi^{02,q}_{e'}$ to end up at $\ket{S}$. Provided that the condition in the lemma statement holds, this can never lead to inconsistency. Ensuring that the condition holds on \textit{subsequent} $\Pi^{02,q}_e$ clauses along $G_q$ is sufficient to ensure the consistency between non-subsequent $\Pi^{02,q}_e$ clauses. 

For case (2), the $\Pi^{02,q}_e$ clauses cannot lead to inconsistency, since there is only one such clause in $G_q$. 

For case (3), the same process as for case (1) determines the consistency conditions. However, there can be a non-loop-closing $\Pi^{02,q}_{(1,n_q)}$ clause (i.e., on edge $(1,n_q)$). Again using Lemma \ref{lemma:partners}, it can be shown straightforwardly that the altered $(\star)$ condition implies consistency between the $\Pi^{02,q}_{(1,n_q)}$ clause and its subsequent $\Pi^{02,q}_{e}$ clause along the loop. 
\end{proof} 

Importantly, we can show that satisfying assignments on non-hPNC clusters which are lines or loops are always fermionic Gaussian states of fixed parity. This is argued using the fact that a $\Pi^{1,q}_e+\Pi^{02,q}_e$ projector is quadratic, see Section \ref{sec:higherrank}, and arguing that one can always add extra clauses onto edges in the cluster so that all edges become of this form $\Pi^{1,q}_e+\Pi^{02,q}_e$. Satisfying assignments must then lie in the null-space of a quadratic fermionic (positive semi-definite) Hamiltonian, hence be a Gaussian state.

\begin{lemma}
    Given a non-hPNC cluster $G_q$ which is a line or loop. The satisfying assignments (if they exist) on this cluster are Gaussian states. 
\label{cor:onlygaussianSAs}
\end{lemma}
\begin{proof}
Let us first consider line or loop clusters (in the particle-hole transformed basis) for which the only $\Pi^{02,q}_e$ clauses are non-loop-closing. There can be multiple such clauses (i.e., cases (1) and (3) in Lemma \ref{lem:nonhPNClinesloops}), or just a single one. In the former case, we assume Eq. \eqref{eq:nonloopclosingcondition} is satisfied (otherwise, the cluster would not be satisfiable).
Let us add a non-loop-closing $\Pi^{02,q}_e$ clause to an edge not yet containing one. If this $\Pi^{02,q}_e$ clause also obeys Eq. \eqref{eq:nonloopclosingcondition}, then the satisfying assignments of the original problem are still satisfying assignments. One can keep adding $\Pi^{02,q}_e$ clauses in this manner until \textit{all} edges in the cluster are of type $\Pi^{1,q}_e+\Pi^{02,q}_e$. In that case, all cluster projectors are quadratic and hence the satisfying assignments are Gaussian states as per observations in Section \ref{sec:char}.

What is left is to consider clusters whose only $\Pi^{02,q}_{e}$ clause is a loop-closing $\Pi^{02,q}_e$ clause (i.e., Figure \ref{fig:linelooptaxonomy}(e) and case (2) in Lemma \ref{lem:nonhPNClinesloops}). 
Importantly, a quantum cluster with exactly the same satisfying assignments can be obtained by deleting the loop-closing $\Pi^{02,q}_e$ clause, and adding an appropriate non-loop-closing $\Pi^{02,q}_e$ clause anywhere along $T_q$. Hence we have obtained a cluster of type (1) in Lemma \ref{lem:nonhPNClinesloops}, for which we have already shown that the lemma statement holds. 
\end{proof}

\subsubsection{Efficient verification of allowed hidden particle number or parity on a quantum cluster}
\label{sec:eff-ver}
The following Lemma captures that verifying what cluster particle numbers and parities are allowed in the cluster-product form assignment, Eq.~\eqref{eq:clusprod}, already restricted by the results in the previous section, is classically efficient.

\begin{lemma}[Cluster checks]
    Given a quantum cluster $G_q = (V_q,E_q)$ consisting of $n_q$ modes and $m_q$ quantum clauses. If $G_q$ is an hPNC cluster, one can check in $O(n_q+m_q)$ time which hidden cluster particle numbers $N_q$ (being the cluster particle number of $K_{B}\ket{\psi}$, with $\ket{\psi}$ a satisfying assignment) correspond to satisfying assignments. If $G_q$ is a non-hPNC cluster, one can check in $O(n_q+m_q)$ time which cluster parities correspond to satisfying assignments.
    \label{lem:ccheck}
\end{lemma}
\begin{proof}
    First, suppose $G_q$ is a hPNC cluster. By Corollary \ref{cor:restricted-PN} point \ref{alwaysemptyfull}, hidden cluster particle numbers $N_q = 0$ and $N_q = n_q$ are {\em always} permitted. 
    
    If the hPNC cluster $G_{q}$ has a vertex of degree at least $3$, then by Corollary \ref{cor:restricted-PN} point \ref{hPNC3}, to see if a single-particle state with $N_{q} = 1$ and/or single-hole state with $N_q = n_q - 1$ are permitted, one simply lists all $m_q$ partnering constraints (Lemma \ref{lemma:partners}) between $n_{q}$ coefficients in the potential satisfying assignment and checks the consistency of these constraints in time $O(n_q+m_q)$. 
      On the other hand, suppose the hPNC cluster $G_{q}$ has no vertex of degree at least $3$. If $G_q$ is a line, all hidden cluster particle numbers are allowed via Lemma \ref{lem:line}, and the corresponding cluster particle states are unique. In case $G_q$ is a loop, then all cluster particle numbers for which the condition in Lemma \ref{lem:loop} is obeyed, are permitted, with the caveat that $N_q = 0$ and $N_q = n_q$ are permitted in any case.   

    Now suppose that $G_q$ is a non-hPNC cluster. If $G_q$ has a vertex of degree at least $3$, then by Corollary \ref{cor:restricted-PN} any assignment is excluded for $n_q\geq 5$, one parity is excluded for $n_q = 4$, and obviously $n_q$ cannot be $\leq 3$. Clearly, one can efficiently verify the existence of an assignment for $n_q=4$. If the non-hPNC cluster $G_q$ does not have a vertex of degree at least $3$, then $G_q$ can be a line or a loop and which (if any) parity is allowed for those clusters can be checked in $O(n_q+m_q)$ time using the conditions in Lemmas \ref{lem:loop} and \ref{lem:nonhPNClinesloops}.
\end{proof}

\section{An efficient classical algorithm for Fermionic 2-SAT (with and without fixed parity)}
\label{sec:fixedparity}

\subsection{Solving Fermionic 2-SAT: proof of Theorem \ref{theorem:parityconservingP}}
\label{sec:solvingfermionic2SAT}
Let us now prove Theorem \ref{theorem:parityconservingP}, which states that \textsc{Fermionic} 2-SAT can be solved in time $O(n+m)$. 
Solving \textsc{Fermionic} 2-SAT comes down to performing certain checks on properties of quantum clusters, and running a (classical) 2-SAT algorithm based on these constraints and the remaining classical clauses:

\begin{proof}[Proof of Theorem \ref{theorem:parityconservingP}]
Due to Proposition \ref{lemma:global} we can restrict the assignment to cluster-product form. We will be constructing a classical 2-SAT instance, which will be used for deciding whether the \textsc{Fermionic} 2-SAT instance is satisfiable. This classical 2-SAT instance is denoted by $\mathcal{C}$. First, let us consider only the constraints from quantum edges. 
\begin{enumerate}
\item \label{hPNC} We first consider all hPNC clusters. For these clusters, Corollary \ref{cor:restricted-PN} point \ref{alwaysemptyfull} implies that there are \textit{always} two classical assignments that satisfy all quantum clauses in the cluster and that are each other's negations (see Figure \ref{fig:F2SATinP} for an illustration). Now consider the set of classical clauses, either internal to the cluster or straddling between the cluster and some classical modes or between the cluster and some other quantum cluster. By Corollary \ref{cor:restricted-PN} point \ref{needstobefree} for $n_q > 2$, obeying any internal classical clauses on the cluster in any assignment with hidden particle number $N_q \neq 0,N_q \neq n_q$ would also allow for two classical assignments $N_q=0,n_q$. Hence wlog we can assume that we have to choose from these classical assignments for any $n_q>2$ cluster when constructing a global satisfying assignment. One can impose this using a set of 2-SAT clauses whose two unique satisfying assignments are these two classical assignments (which are each other's negations): note that we need to apply the particle-hole transformation $K_B$ to obtain these 2-SAT clauses. These 2-SAT clauses, the classical clauses internal to the cluster and any straddling clauses are included in $\mathcal{C}$. The case $n_q=2$ is dealt with separately for convenience. For $n_q=2$, if there is a single clause and no internal classical clause, we do the same as for $n_q>2$. If $n_q=2$ with one $\Pi^{1,q}_e$ clause and one $\Pi^{02,q}_e$ clause, the cluster is non-hPNC and we deal with it in the next point \ref{non-hPNC}. If $n_q=2$ and the projector is rank-3, then it is hPNC, with 2 classical internal clauses and there is at most one satisfying Gaussian, non-classical assignment. For example, the two classical clauses exclude the assignments 00 and 11, allowing only the state $(\alpha a_1^{\dagger}+\beta a_2^{\dagger})\ket{\rm vac}$ orthogonal to some $\Pi^{1,q}_{(1,2)}$ clause. Hence, in this case only, one has to choose this non-classical (Gaussian) assignment and handle the possible classical clauses straddling this cluster via Corollary \ref{cor:restricted-PN} point \ref{needstobefree}: a straddling clause on mode 1 in the cluster, i.e. $(x_1 \vee u)$ or $(\overline{x}_1 \vee u)$, gets replaced by clause $(u)$ with literal $u$ on some classical mode (and similarly for mode 2 in the cluster). These latter clauses are included in $\mathcal{C}$. 
\item \label{non-hPNC} For all quantum clusters that are non-hPNC, we first check via Lemma \ref{lem:ccheck} in $O(n+m)$ time whether there exists a satisfying assignment. Corollary \ref{cor:restricted-PN} point \ref{needstobefree2} says that if there exists a satisfying assignment, there can be no additional internal classical clauses (as the state has support on all occupations with given parity). We handle any straddling classical clauses by replacing $(x_m \vee u_i)$ or $(\overline{x}_m \vee u_i)$, with mode $m$ inside the non-hPNC cluster and classical literal $u_i$, by clauses $(u_i)$. These latter clauses are included in $\mathcal{C}$. The case when the classical clause is straddling between two quantum clusters can be dealt with similarly.
\end{enumerate}

\begin{figure}[t]
\centering
\includegraphics[width=0.7\textwidth]{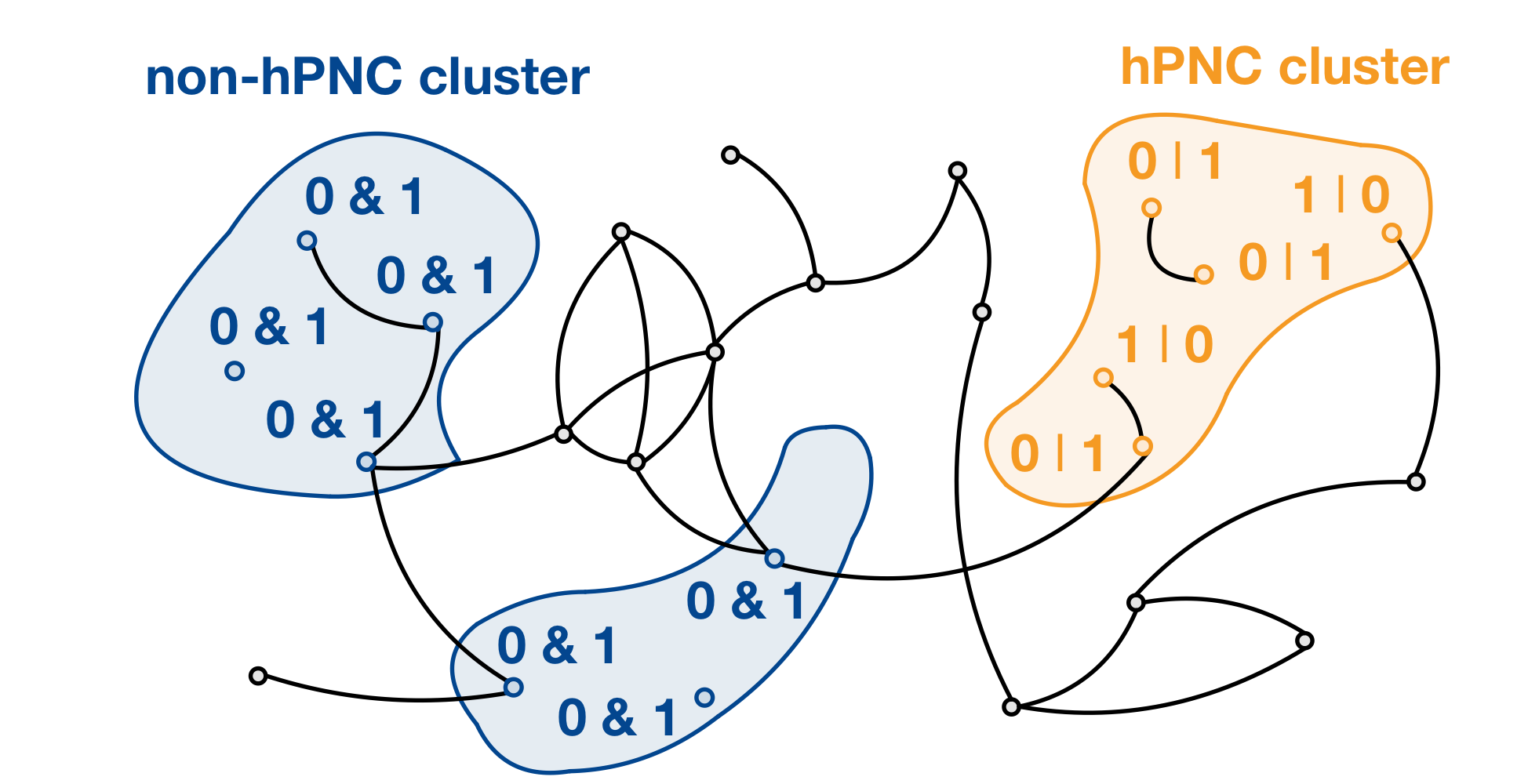}
\caption{Satisfying assignments on non-hPNC quantum clusters are such that each mode has to be both empty \textit{and} occupied, see Corollary \ref{cor:restricted-PN} point \ref{needstobefree2}. On hPNC clusters, one only has to assign either of two classical states to check whether the instance is satisfiable (apart from a special case when the cluster only has two modes, $n_q=2$), when we do not aim to find an assignment of given parity. The black vertices and edges respectively denote modes that are not in any quantum cluster (i.e., classical modes) and classical clauses. }
\label{fig:F2SATinP}
\end{figure}

After these preprocessing steps in which we may find that there can't be a satisfying assignment (then output NO), let the collection of classical clauses obtained via the preprocessing, as well as the remaining clauses on solely classical modes, be the classical 2-SAT instance $\mathcal{C}$. We then solve instance $\mathcal{C}$ in time $O(n+m)$ \cite{Lintimeclassical2sat}.
\end{proof}

\subsection{Solving Fermionic 2-SAT with fixed parity: proof of Theorem \ref{theorem:parityconstrainedP}}

Let us now address the question of whether fixing the parity changes the complexity of the \textsc{Fermionic} 2-SAT problem. We prove Theorem \ref{theorem:parityconstrainedP} in this section, which states that this problem can still be solved efficiently, in time $O(nm)$. Our algorithm for solving \textsc{Fermionic} 2-SAT with fixed parity $P\in \{-1,+1\}$ consists of running certain efficient checks on the allowed satisfying assignments on quantum clusters as a preprocessing step to construct a (classical) 2-SAT with fixed parity instance, on which one runs the $O(nm)$ time algorithm developed in Theorem \ref{thm:parity} in Appendix \ref{app:parityconstrained2SAT}. 

Importantly, non-hPNC clusters are treated differently than hPNC clusters in this next proof. Namely, any assignment in a non-hPNC cluster satisfies the classical cluster-straddling clauses (or not) {\em independent} of its parity, while for hPNC clusters, assignments of different parity can lead to different ways of satisfying the straddling classical clauses.

\begin{proof}[Proof of Theorem \ref{theorem:parityconstrainedP}.]
By Proposition \ref{lemma:global} an assignment of fixed parity is of cluster-product form. 
We proceed in three steps.
\begin{enumerate}
    \item We run the $O(nm)$ time algorithm (developed in Appendix \ref{app:parityconstrained2SAT}) for solving the (classical) 2-SAT instance $\mathcal{C}$ obtained in the proof of Theorem \ref{theorem:parityconservingP} via items \ref{hPNC} and \ref{non-hPNC}, with fixed parity. If the algorithm outputs an assignment of the right parity, we are done and output YES. If there is no assignment, output NO. If there is an assignment but with the wrong parity (which the algorithm will tell you), let's call the assignment $X_0$, and proceed to the next step. 
    \item We consider whether it is possible to flip the parity of any non-hPNC cluster in $X_0$. No variables internal to a non-hPNC cluster enter the instance $\mathcal{C}$ due to how we handle the straddling clauses in item \ref{non-hPNC} in the proof of Theorem \ref{theorem:parityconservingP}. Therefore, if we can flip the parity of a single cluster $G_q^j \in \mathsf{non-hPNC}$, the assignment on $G\backslash G_q^j$ of $X_0$ is still valid, and hence we have full assignment with the correct parity. According to Corollary \ref{cor:allparticlenumberallparity}, each non-hPNC cluster has at most one satisfying assignment per parity. Which parities (one or both) are allowed for each non-hPNC cluster can be checked in time $O(n+m)$ via Lemma \ref{lem:ccheck}. So we iterate over all non-hPNC clusters and if there is one whose parity can be flipped, we are done (output YES). If not, we proceed to the next step.
    \item The next step pertains to all hPNC clusters, except the $n_q=2$ clusters with a rank-3 edge which were treated differently in obtaining $\mathcal{C}$ in item \ref{hPNC} of the proof of Theorem \ref{theorem:parityconservingP}: we will not attempt flipping parity on these clusters as there is only one possible assignment. 
    
    Due to the failure of obtaining a flipped parity assignment in the previous two steps, we know that it is not possible to change the parity by flipping classical modes and/or flipping the hidden particle number $N_q=0$ (all-empty) to $N_q=n_q$ (all-filled) assignment on any hPNC cluster (which are the possible assignments inside the hPNC clusters in $\mathcal{C}$ which are examined in step 1). This has the following consequence. We define the classical 2-SAT instances $\mathcal{C}^{+,\mathsf{M}}$ which is a superset of the classical clauses of $\mathcal{C}$ and some classical clauses which pertain to a subset of non-hPNC clusters $\mathsf{M} \subseteq \mathsf{hPNC}$. That is, $\mathcal{C}^{+,\mathsf{M}}$ is obtained by adding to the collection of clauses in $\mathcal{C}$, for each cluster $G_q^j\in \mathsf{M}$, the following: replace each straddling clause between a mode $i \in G_q^j$ and classical mode, of the form $(x_i \vee u)$ or $(\overline{x}_i \vee u)$ by clause $(u)$ with $u$ a literal of the classical mode. One does the same for any straddling clause with another hPNC quantum cluster provided it is not in the set $\mathsf{M}$ (if it is non-hPNC or in the set $\mathsf{M}$, there would be no satisfying assignment for $\mathcal{C}^{+,\mathsf{M}}$). Due to failure in the previous two steps, either an instance $\mathcal{C}^{+,\mathsf{M}}$ will have no satisfying assignment, or an assignment of the wrong parity, i.e. the same as $X_0$, for any $\mathsf{M}$, since any satisfying assignment of $\mathcal{C}^{+,\mathsf{M}}$ is also a satisfying assignment for $\mathcal{C}$. Hence, the only way to still obtain a flipped parity assignment as compared to $X_0$ is to possibly flip the parity of some individual hPNC cluster by {\em not choosing} the hidden particle number $N_q=0$ or $N_q=n_q$ assignment on it, as we did for $\mathcal{C}$. Note that if we do not constrain the overall parity, as in Theorem \ref{theorem:parityconservingP}, there was never a reason to consider other possible assignments (see also Fig.~\ref{fig:F2SATinP}).
    
    Hence, we first consider for each hPNC cluster $G_q^i$ whether one can flip its parity, and then if so, we decide whether an altered classical 2-SAT instance $\mathcal{C}_{G_q^i}$ has a satisfying assignment using a 2-SAT solver. To check whether one can flip its parity, observe that no internal classical clauses must be present as any assignment with hiddden particle numbers $N_q\neq 0,N_q\neq n_q$ has the property that each mode $j$ in the cluster can be both 0 and 1. If classical clauses are present and/or the parity cannot be flipped for other reasons, we move onto the next hPNC cluster. As for checking on a hPNC cluster whether an assignment of flipped parity is allowed, given its quantum clauses, Lemma \ref{lem:ccheck} shows that this can be done efficiently, in $O(n+m)$ time. For a hPNC cluster which is a line or a loop, in principle non-Gaussian assignments with hidden particle number $N_q\neq 0,1,n_q-1,n_q$ may be allowed. However, one can observe that, if a flipped parity assignment is possible, one can take this assignment actually to be a Gaussian state, i.e. one only with hidden particle number $N_q=0,1,n_q-1,n_q$.
    
    If the parity can be flipped for cluster $G_q^i$, we construct $\mathcal{C}_{G_q^i}$. To get from $\mathcal{C}$ to $\mathcal{C}_{G_q^i}$ we (1) replace all straddling clauses between a mode $j$ in cluster $G_q^i$ and a classical mode $(x_j \vee u)$ and $(\overline{x}_j \vee u)$ by clause $(u)$. One does the same for a straddling clause with a mode in a hPNC quantum cluster (if there is a straddling clause with a non-hPNC quantum cluster, then there won't be a satisfying assignment when flipping cluster $G_q^i$). From $\mathcal{C}$ we remove all classical clauses which constrain the assignment on $G_q^i$ to have hidden particle number $N_q=0$ or $N_q=n_q$ as described in item \ref{hPNC} in the proof of Theorem \ref{theorem:parityconservingP}: there are no more variables in $\mathcal{C}_{G_q^i}$ which pertain to $\mathcal{C}_{G_q^i}$. If $\mathcal{C}_{G_q^i}$ has an assignment, then it must now be an assignment of the right parity and we output YES. If not, for all hPNC clusters that we iterate over, we output NO, and this process is clearly efficient.  
\end{enumerate}
\end{proof}

\section{PNC Fermionic 2-SAT with fixed particle number is NP-complete}
\label{sec:NPC}

Let us now consider the complexity of solving Problem \ref{problem:NparticleFkSAT}, which corresponds to deciding PNC \textsc{Fermionic} 2-SAT with particle number fixed to some $N\in \{0,1,\ldots,n\}$. In contrast to fixing the parity, we show that fixing the particle number to $N$ makes \textsc{Fermionic} 2-SAT NP-complete (Theorem \ref{theorem:NPcomplete}).  

\begin{proof}[Proof of Theorem \ref{theorem:NPcomplete}]
Proposition \ref{lemma:global} and Corollary \ref{cor:restricted-PN} show that satisfying assignments of PNC \textsc{Fermionic} 2-SAT are of cluster-product form. In the PNC version of \textsc{Fermionic} 2-SAT, quantum clauses are only of type $\Pi^{1,q}_{e}$. As a consequence, all quantum clusters are PNC, even without performing a particle-hole transformation. Therefore, we can characterize satisfying assignments on quantum clusters by the particle number on that cluster. A witness $W$ thus consists of a list of $O(n)$ integers $W = \big(\{N_q^{i}\}_{i\in \mathsf{hPNC}},\{x_{k}\}_{k\in \mathsf{Class}}\big)$. 
The verifier algorithm verifies for each quantum cluster $i$, using Lemma \ref{lem:ccheck} whether $N_q^{i}$ is indeed an allowed particle number on that cluster. Note that the verifier does not need to construct or have the actual fermionic state to do this. Using the collection of particle numbers on quantum clusters $N_q^{i}$ and Corollary \ref{cor:restricted-PN}, she infers for each mode in a quantum cluster whether it is empty, occupied or both. Combined with $\{x_{k}\}_{k\in \mathsf{Class}}$, it is then verified whether all classical clauses $e\in E_c$ are indeed satisfied. As a final step, the verifier checks if the total particle number of the witness $W$ is indeed the given $N$. This verification algorithm is efficient. To show NP-hardness, let us consider an instance of \textsc{Fermionic} 2-SAT consisting of monotone classical clauses only, i.e. only clauses of type $(x_{i}\lor x_{j})$. Deciding this problem with fixed Hamming weight $N$ is equivalent to deciding whether there is vertex cover with $N$ vertices, which is NP-complete. Indeed, it is known that classical 2-SAT with fixed Hamming weight, i.e. \textit{the weighted 2-SAT problem} (W2SAT) is NP-complete, see \url{https://en.wikipedia.org/wiki/2-satisfiability}.
\end{proof}

\begin{remark}[Classical Assignments]
Part \ref{hPNC} of the proof of Theorem \ref{theorem:parityconservingP} for hPNC clusters and classical modes shows that PNC \textsc{Fermionic} 2-SAT (which only has $\Pi^{1,q}_e$ quantum clauses and thus hPNC quantum clusters) without any particle number constraint always has a classical satisfying assignment, i.e., a satisfying assignment of the form $a_{S}^{\dagger}\ket{{\rm vac}}$, with the exception of 2-mode quantum clusters with rank-3 projectors in Gaussian states. Note that this is similar to quantum 2-SAT always having a product assignment with the exception of rank-3 edges \cite{ASSZ:linearSAT}. 
In contrast, when fixing particle number to some integer $N$ as in Problem \ref{problem:NparticleFkSAT}, there are instances with only \textit{non-Gaussian} satisfying assignments, like the satisfying assignments on a line for $1 < N_q < n_q-1$ in Lemma \ref{lem:line}.
\end{remark}


\section{Complexity of Fermionic k-SAT and related problems}
\label{sec:kSAT}

Let us provide some background and perspective on the (known) complexity of related fermionic problems. Using qubit-to-fermion mappings, such as the one in Ref.~\cite{HATH:ferm}, we can straightforwardly argue the following, see Appendix \ref{app:qma} for the proof.
\begin{lemma}
\textsc{Fermionic} $k$-\textsc{SAT} $\in$ \textsc{QMA} and \textsc{Fermionic} $k$-\textsc{SAT} is \textsc{QMA}$_1$-hard for $k=9$.
\label{lem:QMA1}
\end{lemma}
\noindent
This Lemma mirrors the results for \textsc{QMA}$_{1}$-hardness of \textsc{Quantum} $k$-SAT for $k\geq 3$ \cite{GN:3SAT, bravyi:quantumSAT}.

One may expect that \textsc{Fermionic} $k$-\textsc{SAT} $\in$ $\textsc{QMA}_1$ for all variants of the problem (particle or non-particle number conserving, with a fixed particle number or not), but one has to be cautious about what basic gates are used in the class $\textsc{QMA}_1$ (usually $H$, CNOT and $T$) and whether this set of gates allows for a zero-error acceptance in the YES case when the fermionic problem and its specification is mapped to qubits.

It may be possible to reduce $k=9$ in Lemma \ref{lem:QMA1} to a lower $k$ by adapting the space-time circuit-to-Hamiltonian construction in Section 3.3 in \cite{BT:spacetime} which proves $\textsc{QMA}$-completeness of a fermionic circuit Hamiltonian with projectors which involve at most 4 fermionic modes ($k=4$), {\em under the restriction} that there is 1 particle per track on the 2D lattice. Thus we don't have an overall particle constraint, but several fixed particle sectors, which can possibly be shown to be equivalent to an overall constraint.

Let us also mention a result that is related to this work. It is known that the `\textsc{Fermionic MAX-}2\textsc{-SAT}' problem {\em with} particle number constraint is QMA-complete, see the next Theorem \ref{theorem:QMA}. Note that this problem is fundamentally different from \textsc{Fermionic }2\textsc{-SAT}, since the former is a ground-state energy estimation problem and not a question of whether a given Hamiltonian is frustrated or not. 

\begin{theorem}[Theorem 2 in \cite{gorman_QMA}]
Determining the ground state energy with $1/{\rm poly}(n)$ precision for a class of Fermi-Hubbard Hamiltonians with $n$ fermionic modes, with particle number $N=n/2$ (half-filling) is $\textsc{QMA}$-complete. The particle-number-conserving Fermi-Hubbard Hamiltonian on a graph $G=(V,E)$ is
\begin{equation}
H_{\rm FH}=U\sum_{i \in V} n_{i,+} n_{i,-}+\sum_{(i,j)\in E}\sum_{\sigma=\pm} t_{i,j} (a_{i,\sigma}^{\dagger} a_{j,\sigma}+a_{j,\sigma}^{\dagger} a_{i,\sigma}),
\end{equation}
with $n_{i,+}=a_{i,+}^{\dagger} a_{i,+}$, $U,t_{i,j}\in \mathbb{R}$. For $\textsc{QMA}$-completeness, bounds are specified on the parameters $U$ and $\{t_{i,j}\}$. 
\label{theorem:QMA}
\end{theorem}
\noindent 

\section{Discussion} 
\label{sec:discuss}

Interestingly, it is not clear whether there is an efficient classical algorithm to solve instances of \textsc{Quantum} 2-SAT (with only parity-conserving projectors) when we ask for an assignment with fixed parity: in this case there can be non-product satisfying assignments. In addition, we don't have the ``natively-fermionic'' Lemma \ref{lemma:degreeatleast3} for \textsc{Quantum} 2-SAT. We conjecture that this problem is of different complexity than \textsc{Fermionic} 2-SAT with fixed parity which we have proved to be efficiently solvable in Theorem \ref{theorem:parityconstrainedP}. This would mesh elegantly with the point of view that parity-conserving interactions are fundamental for fermions. 

We have seen that the satisfying assignments of \textsc{Fermionic} $2$-SAT are of cluster-product form with some modes with purely classical occupations, some modes in a possibly-large clusters in Gaussian states of fixed parity or hidden fixed particle number, and some modes in a 4-fermion cluster in a non-Gaussian state. It will be interesting to explore how we can use these results for \textsc{Fermionic} 2-SAT to develop approximation algorithms or heuristic strategies, quantum or classical, to solve \textsc{Fermionic} $k$-SAT for $k> 2$. For \textsc{Fermionic} $(k>2)$-SAT one may expect that satisfying assignments can be genuinely many-mode non-Gaussian states: due to its QMA$_1$-hardness, satisfying assignments for general \textsc{Fermionic} $k$-SAT problems are not expected to be classically efficiently describable. Thus for such problems, one can seek quantum heuristic strategies, i.e. a quantum equivalent of classical heuristic SAT solvers, which aim at constructing a satisfying state on a quantum computer: such strategies could build on the nature of satisfying assignments for the (\textsc{Fermionic}) 2-SAT problem. It is an open question whether there are interesting classical mathematical problems which can be formulated as a question about the existence of a satisfying assignment to a finite-size, quantum or \textsc{Fermionic} $k$-SAT problem,---perhaps the results in \cite{KK:cliquehomo} can be useful here---. Such construction would be a quantum counterpart to a classical computer-assisted proof obtained through the use of a classical heuristic SAT solver \cite{heule:SAT}.

As a general question, it might be interesting to consider \textsc{Fermionic} $(k>2)$-SAT problems with additional fermionic symmetry, for example consider whether there are complexity-theoretic consequences of time-reversal, spatial-parity and charge-conjugation symmetries as used in the classification of non-interacting fermionic models \cite{Kitaev_2009}.

\section*{Acknowledgements}
We thank Y.~Herasymenko and M.H.~Shaw for insightful discussions. This work is supported by QuTech NWO funding 2020-2024 – Part I “Fundamental Research”, project number 601.QT.001-1, financed by the Dutch Research Council (NWO).

\bibliographystyle{quantum}

%

\appendix

\section{Mathematical micro-facts for Section \ref{section:prelims}}
\label{app:prelims}

\subsection{Invariance of Fermionic 2-SAT clauses that exclude all-empty or all-filled states}

For an edge $e=(j,k)$ one can define rotated annihilation operators:
\begin{equation}
\begin{pmatrix}
\tilde{a}_{j} \\
\tilde{a}_{k} 
\end{pmatrix}
= 
\begin{pmatrix}
U_{11} & U_{12} \\
U_{21} & U_{22}  
\end{pmatrix}
\begin{pmatrix}
a_{j} \\
a_{k} 
\end{pmatrix},
\label{eq:jktransformation}
\end{equation}
with $U$ a unitary matrix.

We can define
\begin{equation}
\Pi^0_{e}= \big(I-\tilde{a}_{j}^{\dagger}\tilde{a}_{j}\big)\big(I-\tilde{a}_{k}^{\dagger}\tilde{a}_{k}\big), 
\label{eq:pi0'}
\end{equation}
and
\begin{equation}
\Pi^2_{e}= \tilde{a}_{j}^{\dagger}\tilde{a}_{j} \tilde{a}_{k}^{\dagger}\tilde{a}_{k}.
\label{eq:pi2'}
\end{equation}

The following holds
\begin{proposition}
Projectors $\Pi^0_e$ and $\Pi^2_e$ in Eqs.~\eqref{eq:pi0'} and \eqref{eq:pi2'} are invariant under any transformation $U$, as in Eq.~\eqref{eq:jktransformation}, and can thus be viewed as classical clauses in the $\{a_{j},a_{k}\}$ mode basis, with the classical Boolean variables corresponding to occupation numbers of the $\{a_{j},a_{k}\}$ modes. That is, $\Pi_{e}^2$ becomes the clause $(\overline{x}_j \vee  \overline{x}_k)$ (excluding 11) and $\Pi_{e}^0$ becomes $(x_j \vee x_k)$ (excluding 00) as in Eq.~\eqref{eq:pi0andpi2}.
\label{prop:02invariance}
\end{proposition}

\begin{proof} 
Inserting the transformation $U$ into $\Pi^{0}_{e}$ with $e = (j,k)$ gives
\vspace{-0.7cm}
\begin{center}
\begin{align}
\Pi^0_{e}=\: \tilde{a}_{j}\tilde{a}^{\dagger}_{j}\tilde{a}_{k}\tilde{a}^{\dagger}_{k} \nonumber \\ 
= \big[I -\big(U^{*}_{11}a_{j}^{\dagger}+U^{*}_{12}a_{k}^{\dagger}\big)\big(U_{11}a_{j} + U_{12}a_{k}\big) \big]\big[I -\big(U^{*}_{21}a_{j}^{\dagger}+U^{*}_{22}a_{k}^{\dagger}\big)\big(U_{21}a_{j} + U_{22}a_{k}\big) \big] \nonumber \\ 
=I - a_{j}^{\dagger}a_{j} - a_{k}^{\dagger}a_{k} + |\det U\:|^2 a_{j}^{\dagger}a_{j}a_{k}^{\dagger}a_{k} =
a_{j}a_{j}^{\dagger}a_{k}a_{k}^{\dagger},    
\end{align}
\end{center}
where we have used unitarity of $U$: $U_{21} U_{22}^*=-U_{11} U_{12}^*$ and $|\det U\:| = 1$. Similarly, since $\Pi_e^2=I-\tilde{a}_{j}^{\dagger}\tilde{a}_{j}-\tilde{a}_{k}^{\dagger}\tilde{a}_{k}+\Pi^0_e$, and $\tilde{a}_{j}^{\dagger}\tilde{a}_{j}+\tilde{a}_{k}^{\dagger}\tilde{a}_{k}$ is invariant under $U$, $\Pi_e^2$ is also invariant under $U$: $\Pi_e^2 = \tilde{a}_{j}^{\dagger}\tilde{a}_{j}\tilde{a}_{k}^{\dagger}\tilde{a}_{k} =  a_{j}^{\dagger}a_{j}a_{k}^{\dagger}a_{k}$.
\end{proof}

\subsection{Action of the particle-hole transformation $K_S$ on $\Pi^{1,q}_e$ and $\Pi^{02,q}_e$ clauses}
\label{sec:PH-details}

The transformation rules are
\begin{align}
    j\in S,\:k\notin S \colon &    \left\{ \begin{aligned} 
 K_{S}\Pi^{1,q}_{e}K_{S}^{-1} &= \tilde{\Pi}^{02,q}_{e} \text{ with } \tilde{\alpha}_{e} = -\bigg[\hspace{-0.1cm} \prod_{\substack{j\leq i < k \\ \text{ s.t. } i\in S}} \hspace{-0.3cm} (-1) \hspace{0.1cm}\bigg] \hspace{0.1cm} \beta_{e} \text{ and } \tilde{\delta}_{e} = \gamma_{e} 
  \\
  K_{S}\Pi^{02,q}_{e}K_{S}^{-1} &= \tilde{\Pi}^{1,q}_{e} \text{ with } \tilde{\beta}_{e} = -\bigg[\hspace{-0.1cm} \prod_{\substack{j\leq i < k \\ \text{ s.t. } i\in S}} \hspace{-0.3cm} (-1) \hspace{0.1cm}\bigg] \hspace{0.1cm} \alpha_{e} \text{ and } \tilde{\gamma}_{e} = \delta_{e}. 
\end{aligned} \right. \\
    j\notin S,\:k\in S  \colon &  \left\{ \begin{aligned} 
 K_{S}\Pi^{1,q}_{e}K_{S}^{-1} &= \tilde{\Pi}^{02,q}_{e} \text{ with } \tilde{\alpha}_{e} = \bigg[\hspace{-0.1cm} \prod_{\substack{j\leq i < k \\ \text{ s.t. } i\in S}} \hspace{-0.3cm} (-1) \hspace{0.1cm}\bigg] \hspace{0.1cm} \gamma_{e} \text{ and } \tilde{\delta}_{e} = \beta_{e}.
  \\
  K_{S}\Pi^{02,q}_{e}K_{S}^{-1} &= \tilde{\Pi}^{1,q}_{e} \text{ with } \tilde{\beta}_{e} = \bigg[\hspace{-0.1cm} \prod_{\substack{j\leq i < k \\ \text{ s.t. } i\in S}} \hspace{-0.3cm} (-1) \hspace{0.1cm}\bigg] \hspace{0.1cm} \delta_{e} \text{ and } \tilde{\gamma}_{e} = \alpha_{e}.
\end{aligned} \right. \\
 j\in S,\:k\in S \colon &\left\{ \begin{aligned}
  K_{S}\Pi^{1,q}_{e}K_{S}^{-1} &= \tilde{\Pi}^{1,q}_{e} \text{ with } \tilde{\beta}_{e} = -\bigg[\hspace{-0.1cm} \prod_{\substack{j\leq i < k \\ \text{ s.t. } i\in S}} \hspace{-0.3cm} (-1) \hspace{0.1cm}\bigg] \hspace{0.1cm} \gamma_{e} \text{ and } \tilde{\gamma}_{e} = \beta_{e}. \\
 K_{S}\Pi^{02,q}_{e}K_{S}^{-1} &= \tilde{\Pi}^{02,q}_{e} \text{ with } \tilde{\alpha}_{e} = -\bigg[\hspace{-0.1cm} \prod_{\substack{j\leq i < k \\ \text{ s.t. } i\in S}} \hspace{-0.3cm} (-1) \hspace{0.1cm}\bigg] \hspace{0.1cm} \delta_{e} \text{ and } \tilde{\delta}_{e} = \alpha_{e}.
\end{aligned} \right. \\
  j\notin S,\:k\notin S\colon & \left\{ \begin{aligned}
  K_{S}\Pi^{1,q}_{e}K_{S}^{-1} &= \tilde{\Pi}^{1,q}_{e} \text{ with } \tilde{\beta}_{e} = \bigg[\hspace{-0.1cm} \prod_{\substack{j\leq i < k \\ \text{ s.t. } i\in S}} \hspace{-0.3cm} (-1) \hspace{0.1cm}\bigg] \hspace{0.1cm} \beta_{e} \text{ and } \tilde{\gamma}_{e} = \gamma_{e}. \\
 K_{S}\Pi^{02,q}_{e}K_{S}^{-1} &= \tilde{\Pi}^{02,q}_{e} \text{ with } \tilde{\alpha}_{e} = \bigg[\hspace{-0.1cm} \prod_{\substack{j\leq i < k \\ \text{ s.t. } i\in S}} \hspace{-0.3cm} (-1) \hspace{0.1cm}\bigg] \hspace{0.1cm} \alpha_{e} \text{ and } \tilde{\delta}_{e} = \delta_{e}.
\end{aligned} \right. 
    \label{eq:ph-edge}
\end{align}

\section{Example of a unique 4-fermion non-Gaussian, non-product satisfying assignment}
\label{app:exampleNG}

We work through explicitly how the satisfying assignments depend on fermionic parity for a simple illustrative 4-fermionic problem where one has a line of three $\Pi^{1,q}_e$ clauses on modes 1,2,3 and 4 as in Lemma \ref{lem:line}, and one adds a single $\Pi^{02,q}_{(2,4)}$ clause between modes 2 and 4. This is an example of a non-hPNC cluster with vertex 2 having degree 3, i.e. an example of Fig.~\ref{fig:linelooptaxonomy}(d).

This example gives insight into why product state assignments suffice for \textsc{Quantum} 2-SAT (with the exception of cases involving rank-3 projectors) but 4-fermion non-product, non-Gaussian states are needed for \textsc{Fermionic} 2-SAT. On a separate note: from Lemma \ref{lem:line} it is clear that if we ask for fixed particle number for \textsc{Quantum} 2-SAT, one \textit{can} have a unique satisfying assignment which is not a product state, i.e. a fully permutation-symmetric state of $n_q$ qubits with a fixed number $N_q$ of $1$s (with $N_q$ unequal to $0$ or the maximum $n_q$).

The instance we consider here consists of 4 fermionic modes with projectors $\Pi^{1,q}_{(1,2)}$, $\Pi^{1,q}_{(2,3)}, \Pi^{1,q}_{(3,4)}$ which we give in qubit language. Specifically, we take $\Pi_{(1,2)}^{\rm qubit}=\ket{\Psi^-}\bra{\Psi^-}_{12}$, $\Pi_{(2,3)}^{\rm qubit}=\ket{\Psi^-}\bra{\Psi^-}_{23}$,$\Pi_{(3,4)}^{\rm qubit}=\ket{\Psi^-}\bra{\Psi^-}_{34}$, with $\ket{\Psi^-}$ denoting a singlet state. For a fixed particle number $0\leq N\leq 4$, the permutation symmetric state with $N$ excitations is the unique satisfying assignment (i.e., ground state of the ferromagnetic Heisenberg model) via Lemma \ref{lem:line}. For each $N$, there is thus a satisfying assignment, i.e.
\begin{align}
\ket{\psi_{N=0}}=&\: \ket{0000} \nonumber \\ 
\ket{\psi_{N=1}}=&\: \ket{0001}+\ket{0010}+\ket{0100}+\ket{1000}, \nonumber \\ 
\ket{\psi_{N=2}}=&\: \ket{1100}+\ket{0011}+\ket{1010}+\ket{1001}+\ket{0110}+\ket{0101} \nonumber \\ 
\ket{\psi_{N=3}}=&\: \ket{0111}+\ket{1011}+\ket{1101}+\ket{1110}, \nonumber \\ 
\ket{\psi_{N=4}}=&\: \ket{1111}.
\end{align}

Now image we wish to add a fourth fermionic $\Pi^{02,q}_{(2,4)}$-type projector with amplitudes $\alpha_{(2,4)} = \delta_{(2,4)} = 1/\sqrt{2}$ which, in qubit language ---note the additional $Z_3$ due to the Jordan-Wigner transformation in Eq.~\eqref{eq:JW}---, equals
\begin{align}
    \Pi_{(2,4)}^{\rm qubit} =&\: \frac{1}{2}\Big(\sigma^-_{2}\sigma^+_{2}\sigma^-_{4}\sigma^+_{4} + \sigma^-_{2} Z_3 \sigma^-_{4} + \sigma^+_{2}  Z_3\sigma^+_{4} + \sigma_{2}^+\sigma^-_{2}\sigma_{4}^{+}\sigma^-_{4}\Big) \notag \\ =&\:\frac{1}{2}\Big( \ket{00}\bra{00}_{24}+Z_3 [\ket{00}\bra{11}_{24} + \ket{11}\bra{00}_{24}] + \ket{11}\bra{11}_{24}\Big).
    \label{eq:pi24}
\end{align}
From the expression for $\Pi_{(2,4)}^{\rm qubit}$, it is clear that an assignment $\ket{\psi}$ can only be a satisfying assignment if it is of the form $\ket{\psi_{\rm even}} = a_{0}\ket{\psi_{N=0}} + a_{2}\ket{\psi_{N=2}} + a_{4}\ket{\psi_{N=4}}$ (even parity) or $\ket{\psi_{\rm odd}} = a_{1}\ket{\psi_{N=1}} + a_{3}\ket{\psi_{N=3}}$ (odd parity). Let us see which (if any) of these states is projected to zero by $\Pi_{(2,4)}^{\rm qubit}$.
\begin{align}
    \Pi_{(2,4)}^{\rm qubit}\ket{\psi_{\rm even}} =&\: \frac{1}{2}\Big( a_0(\ket{0000}+\ket{0101}) + a_2(\ket{1010}-\ket{1111}+\ket{0000}+\ket{0101}) \nonumber \\ &\: \hspace{7.0cm} + a_4(-\ket{1010}+\ket{1111}) \Big), \nonumber \\
    \Pi_{(2,4)}^{\rm qubit}\ket{\psi_{\rm odd}} =&\: \frac{1}{2}\Big(a_1(\ket{0010}-\ket{0111}+\ket{1000}+\ket{1101}) \nonumber \\ &\: \hspace{4.0cm} + a_3(-\ket{0010}+\ket{0111}+\ket{1000}+\ket{1101})\Big),
\end{align}
where the different signs are caused by $Z_3$ in $\Pi_{(2,4)}^{\rm qubit}$. Clearly, $\Pi_{(2,4)}^{\rm qubit}\ket{\psi_{\rm even}} = 0$ for $a_0 = -a_2$ and $a_2 = a_4$, and $\Pi_{(2,4)}^{\rm qubit}\ket{\psi_{\rm odd}}$ cannot be zero. So, $\ket{\psi_{\rm even}}$ with $a_0 = -a_2$ and $a_2 = a_4$ is the unique satisfying assignment of $\{\Pi_{(1,2)}^{\rm qubit},\Pi_{(2,3)}^{\rm qubit},\Pi_{(3,4)}^{\rm qubit},\Pi_{(2,4)}^{\rm qubit}\}$, which is clearly not a product state. Let us set $a_0 = 1/2\sqrt{2}$, $a_2 = -1/2\sqrt{2}$ and $a_4 = -1/2\sqrt{2}$ wlog (so that $\ket{\psi_{\rm even}}$ is normalized). 

Back in fermionic language, we thus have a satisfying assignment
\begin{align}
    \ket{\psi_{\rm even,f}}=\frac{1}{2\sqrt{2}}\left(I -(a_1^{\dagger}a_2^{\dagger}+ a_3^{\dagger}a_4^{\dagger}+a_2^{\dagger} a_3^{\dagger} +a_1^{\dagger} a_4^{\dagger}+a_2^{\dagger} a_4^{\dagger}+a_1^{\dagger}a_3^{\dagger})-a_1^{\dagger}a_2^{\dagger}a_3^{\dagger}a_4^{\dagger}\right)\ket{\rm vac}.
 \label{eq:SATNG}
\end{align}
To prove that this is a non-Gaussian state, we argue as follows. A pure Gaussian state has a correlation matrix $M\in \mathbb{R}^{8\times 8}$ (with entries in Eq.~(\ref{eq:cor-mat})) with orthonormal columns as $M^{T}M = I$. To show that this does not hold for our state, let's evaluate the first column of $M$ ($M_{j,1}$ for $j = 1,2,\ldots,8$).

By definition, we have that $M_{1,1} = 0$. Furthermore, $M_{2,1} = 2\bra{\psi_{\rm even,f}}a_{1}^{\dagger}a_{1}\ket{\psi_{\rm even,f}}-1$, which can be simply seen to equal $0$. For the other entries, let us distinguish between 
\begin{align}
    &M_{2j-1,1} = i\bra{\psi_{\rm even,f}}(a_{j}+a_{j}^{\dagger})(a_{1} + a_{1}^{\dagger})\ket{\psi_{\rm even,f}},\text{ with odd index $2j-1$,} \nonumber \\
    &M_{2j,1} = i\bra{\psi_{\rm even,f}}i(a_{j}-a_{j}^{\dagger})(a_{1} + a_{1}^{\dagger})\ket{\psi_{\rm even,f}},\text{ with even index $2j$,} 
\end{align}
for $j = 2,3,4$. For either type of entry, we have to evaluate the following expectation values, which can be done straightforwardly. 
\begin{align}
    &\bra{\psi_{\rm even,f}}a_{j}a_{1}^{\dagger}\ket{\psi_{\rm even,f}} = 
    \begin{cases}
      -1/4 & \text{if $j=2$,}\\
      0 & \text{if $j=3$,}\\
      1/4 & \text{if $j=4$.}
    \end{cases} 
    \hspace{0.8cm}
    \bra{\psi_{\rm even,f}}a_{j}^{\dagger}a_{1}\ket{\psi_{\rm even,f}} = 
    \begin{cases}
      1/4 & \text{if $j=2$,}\\
      0 & \text{if $j=3$,}\\
      -1/4 & \text{if $j=4$.}
    \end{cases}
    \nonumber \\
    &\bra{\psi_{\rm even,f}}a_{j}^{\dagger}a_{1}^{\dagger}\ket{\psi_{\rm even,f}} = 
    \begin{cases}
      0 & \text{if $j=2$,}\\
      1/4 & \text{if $j=3$,}\\
      0 & \text{if $j=4$.}
    \end{cases}
    \hspace{1.1cm}
    \bra{\psi_{\rm even,f}}a_{j}a_{1}\ket{\psi_{\rm even,f}} = 
    \begin{cases}
      0 & \text{if $j=2$,}\\
      -1/4 & \text{if $j=3$,}\\
      0 & \text{if $j=4$.}
    \end{cases}
\end{align}
Using these expressions, we conclude the following. 
\begin{equation}
    M_{2j-1,1} = 0 \text{ for $j=2,3,4$ } \hspace{0.4cm} \text{and} \hspace{0.4cm} M_{2j,1} = 
    \begin{cases}
      1/2 & \text{if $j=2$,}\\
      1/2 & \text{if $j=3$,}\\
      -1/2 & \text{if $j=4$.}
    \end{cases}
\end{equation}
Thus the first column of $M$ is $(0,0,0,1/2,0,1/2,0,-1/2)^{T}$, which has 2-norm $\sqrt{3/4}<1$, so $\ket{\psi_{\rm even,f}}$ is not Gaussian. 

Let's compare our findings briefly with an equivalent set of \textsc{Quantum} 2-SAT constraints: Will there be a product assignment here? The \textsc{Quantum} 2-SAT equivalent of the \textsc{Fermionic} 2-SAT instance considered here is $\Big\{\Pi_{(1,2)}^{\rm qubit} = \ket{\Psi^-}\bra{\Psi^-}_{12}$, $\Pi_{(2,3)}^{\rm qubit} = \ket{\Psi^-}\bra{\Psi^-}_{23}$, $\Pi_{(3,4)}^{\rm qubit} = \ket{\Psi^-}\bra{\Psi^-}_{34}$, $\Pi_{(2,4)}^{\rm qubit} = \frac{1}{2}(\ket{00}+\ket{11})(\bra{00}+\ket{11})_{24}\Big\}$ (note the absence of the Jordan-Wigner $Z_3$). This \textsc{Quantum} 2-SAT instance indeed has a product state satisfying assignment, namely $\bigotimes_{j=1}^{4}\frac{1}{\sqrt{2}}(\ket{0}+i\ket{1})_{j}$.

Our pedestrian findings here are captured by the statement for $n_q=4$ in Corollary \ref{cor:restricted-PN} point \ref{degree3}. The degree-3 vertex in this case is vertex 2, the particle-hole-transformation acts on modes $1,3,4$ and changes the parity of the non-Gaussian even parity state in Eq.~\eqref{eq:SATNG} to an odd parity state, while an even parity state in the particle-hole transformed basis (hence odd particle here) is disallowed. The exclusion of one of the two parities for this \textsc{Fermionic} 2-SAT instance is essentially what causes the unique satisfying assignment to be non-product and non-Gaussian. 

\section{Classical 2-SAT with fixed parity}
\label{app:parityconstrained2SAT}
In this section, we prove that classical 2-SAT can be solved efficiently, even when constraining the Hamming weight parity of the assignment. This is not necessarily a trivial problem as they are, for example, simple examples of 2-SAT instances with exponentially many solutions, \textit{all} with the same parity. Consider for instance 
\begin{equation}
    \bigwedge_{j=0}^{n/4-1} (x_{j+1}\lor \overline{x}_{j+2}) \land (x_{j+2}\lor \overline{x}_{j+3}) \land (x_{j+3}\lor \overline{x}_{j+4}) \land (\overline{x}_{j+1}\lor x_{j+4}),
\end{equation}
with $n$ a multiple of four. This instance corresponds to disjoint units of four variables, where the only satisfying assignments on each unit are the all-zeros or the all-ones assignment. Clearly, there are exponentially many satisfying assignments, and they are all of even parity. In fact, there are even simple examples of \textit{connected} 2-SAT instances with exponentially many satisfying assignments, all with the same parity. Thus even if one is guaranteed that a 2-SAT instance has (exponentially) many satisfying assignments, the question of whether there is a satisfying assignment with a given parity is a nontrivial one. We prove the following

\begin{theorem}
      Given an instance of classical 2-SAT on $n$ variables with $m$ clauses, and a parity $P\in \{-1,+1\}$. Decide whether there exists a satisfying assignment $x$ with Hamming weight parity $P$ (YES), or there is no satisfying assignment with Hamming weight parity $P$ (NO). This problem can be decided in time $O(nm)$. 
    \label{thm:parity}
\end{theorem} 

\begin{proof}
    First, we find a satisfying assignment in time $O(n+m)$ \cite{Lintimeclassical2sat} for the $2$-SAT instance, if it is satisfiable. If this solution has Hamming weight parity $P$, then we are done (output YES), else we proceed. For convenience, let us redefine the $2$-SAT instance such that the obtained solution is the all-zeros assignment. Then what is left is to check whether there exists an odd Hamming weight satisfying assignment for the redefined instance. The redefined instance wlog consists of clauses $(x_{i}\lor \overline{x}_{j})$ (with $i<j$ or $i>j$) and $(\overline{x}_{i}\lor \overline{x}_{j})$ (with $i<j$ wlog). Note that there can be multiple clauses per pair of variables $i,j$. 

    Let us consider the sub-graph $G_{\rm sub}$ consisting of just the clauses of type $(x_{i}\lor \overline{x}_{j})$ (with $i<j$ or $i>j$). We associate a directed edge $i\to j$ with each clause $(x_{i}\lor \overline{x}_{j})$ (i.e., edges point from ``variable'' $i$ to ``negated variable'' $j$). Note that two variables $i,j$ can simultaneously be connected by a $(x_{i}\lor \overline{x}_{j})$ clause and a $(x_{j}\lor \overline{x}_{i})$ clause, resulting in a $i\to j$ edge and a $j\to i$ edge in $G_{\rm sub}$. Next, we identify the strongly-connected components (SCC's) of the directed graph $G_{\rm sub}$ in time $O(n+m)$ \cite{TarjanSCC}. The only satisfying assignments on these SCC's are the all-zeros and the all-ones assignments, see Lemma \ref{lemma:SCCassignemnts} below. Hence we compress each SCC into a single Boolean variable. We label the new collection of variables by $\{y_{j}\}_{j=1}^{\leq n}$, where some were previously single variables and others are compressed SCC's. Since each directed cycle is (part of) an SCC, the compressed problem cannot contain any directed cycles, i.e., it is a directed acyclic graph. The graph with $(y_{i}\lor \overline{y}_{j})$-type edges can be topologically sorted in time $O(n+m)$ \cite{TarjanTopologicalSorting}. The topologically-sorted graph now only includes edges of type $(y_{i}\lor \overline{y}_{j})$ for which $i<j$. Next, let us add the $(\overline{x}_{i}\lor \overline{x}_{j})$-type clauses, which after compression have become either $(\overline{y}_{i}\lor \overline{y}_{j})$ clauses or self-edges $(\overline{y}_{i}\lor \overline{y}_{i})$. The former can be added s.t. $i<j$ in the topologically sorted graph (since, obviously, $(\overline{y}_{i}\lor \overline{y}_{j})$ and $(\overline{y}_{j}\lor \overline{y}_{i})$ are equivalent) to construct the new graph $\tilde{G}$, see Fig.~\ref{lemma:SCCassignemnts}. The self-edges of type $(\overline{y}_{i}\lor \overline{y}_{i})$ exclude $y_{i}=1$. These variables with self-edges can be flagged in time $O(n+m)$. 

    With each variable $j$ in the compressed problem, we associate a weight $w_{j}\in \{0,1\}$. The weight $w_{j}$ corresponds to the number (modulo $2$) of original variables in the SCC that has been compressed into variable $y_{j}$. Note that for any uncompressed variable, we assign weight $1$. To obtain a satisfying assignment of odd Hamming weight (if it exists), we need to flip an odd number of odd-weight variables. In order to do so, we might have to simultaneously flip some even-weight variables. 

    Importantly, if we flip a variable $j$ (with some weight $w_{j}$) to $y_{j} = 1$, then, by construction, all clauses to variables in \textit{later} layers of the topological sort are still satisfied. Clauses from variable $j$ to variables in \textit{earlier} layers can now become unsatisfied and thus some variables $i$ s.t. $i<j$ might also have to be set to $y_{i} = 1$. Similarly, flipping such a variable $i$ can only lead to unsatisfied clauses that connect $i$ to variables in earlier layers of the topological sort, not to variables in later layers.

    The algorithm runs by first identifying the earliest odd-weight variable $j$ in the topological sort. For some instances, this \textit{might} already be in the first layer. We set $y_{j} = 1$ and see whether the propagation to earlier layers of the topological sort does not lead to any contradiction. If it does not, then we have found a satisfying assignment of flipped parity (output YES), since all variables that are flipped in earlier layers due to propagation are even-weight variables by definition. If the propagation leads to contradiction, then clearly there is \textit{no} satisfying assignment for which $y_{j} = 1$. Next, we reset to the all-zeros assignment and identify the next odd-weight variable $k$ in the topological sort (which is in the same or in a later layer). Again, one checks whether setting $y_{k} = 1$ leads to contradiction in earlier layers. If it does not, then we have constructed a satisfying assignment with flipped parity (output YES), since $y_{k}$ is the \textit{only} odd-weight variable set to $1$ in that assignment. Indeed, if $y_{j}$ would also have to be set to $1$, then there would be a contradiction. If setting $y_{k}=1$ does lead to contradiction, then there is \textit{no} satisfying assignment for which $y_{k} = 1$. We proceed again by resetting to the all-zeros assignment and identifying the next odd-weight variable in the topological sort. We iterate over all odd-weight variables in this manner and either obtain a satisfying assignment of flipped parity or conclude that no odd-weight variable can be set to $1$ consistently and hence there is no satisfying assignment with flipped parity (output NO). Iterating over all odd-weight variables and checking for the consistency of flipping them takes time $O(nm)$. 
\end{proof}

\begin{figure}[H]
\centering
\includegraphics[width=0.6\textwidth]{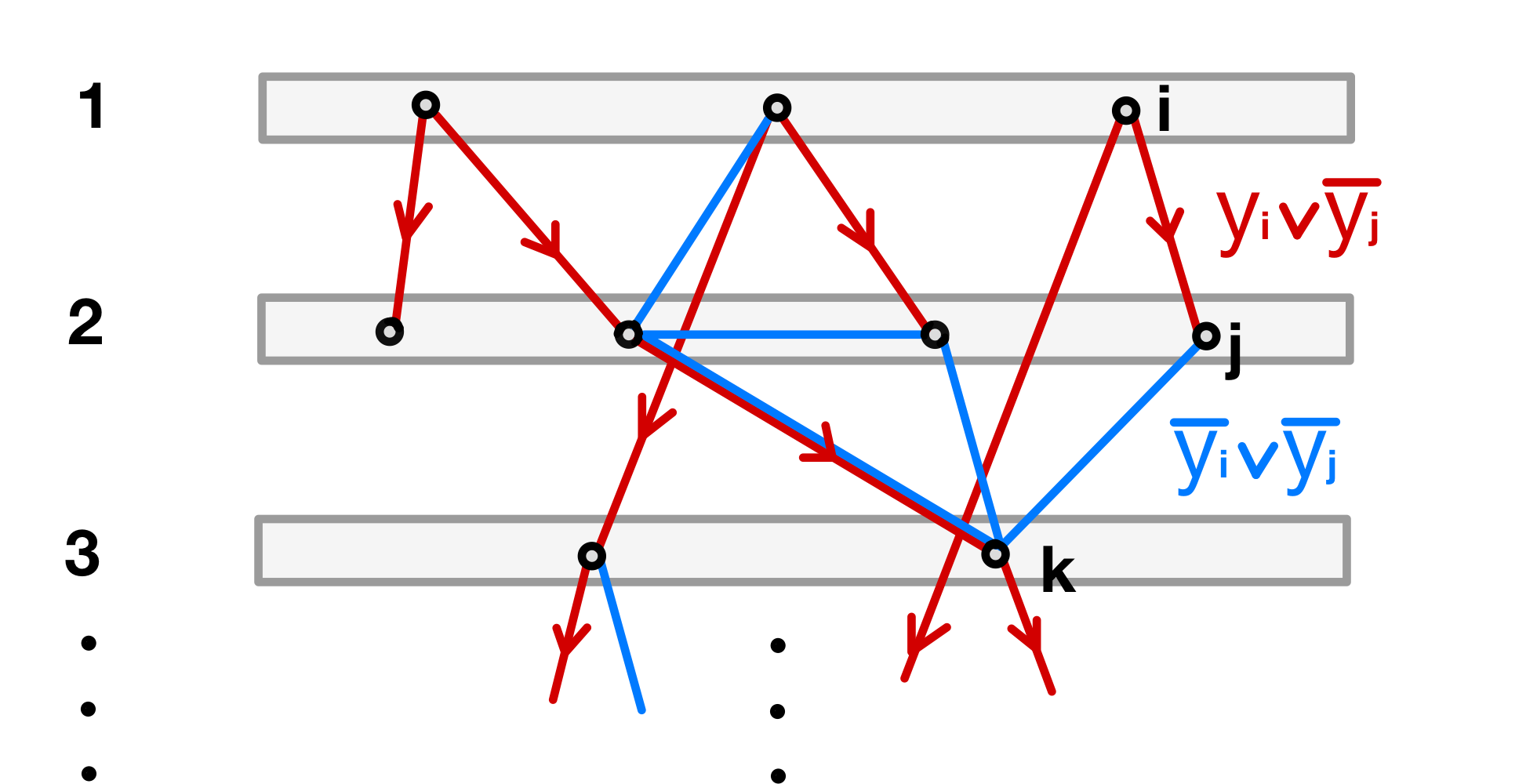}
\caption{First three layers of a compressed topologically sorted graph $\tilde{G}$ derived from $G$, where each variable $j$ has been given a weight $w_{j}$. The red clauses are of type $(y_{i}\lor \overline{y}_{j})$ (with $i<j$) and the blue clauses are of type $(\overline{y}_{i}\lor \overline{y}_{j})$. }
\label{fig:TopolSort}
\end{figure}

\begin{lemma}
    Given an instance of classical 2-SAT with $(x_{i}\lor \overline{x}_{j})$-type clauses only. Let us consider the directed graph with edges $i\to j$ for each clause $(x_{i}\lor \overline{x}_{j})$ (i.e., edges pointing from ``variable'' $i$ to ``negated variable'' $j$). If this graph is strongly connected, then only the all-zeros and all-ones assignments are satisfying assignments. 
    \label{lemma:SCCassignemnts}
\end{lemma} 

\begin{proof}
    If the directed graph is strongly connected, then for \textit{each} pair of variables $(i,j)$ there is path from $i$ to $j$ and a path from $j$ to $i$. Let us set $x_{i} = 0$. Then clauses along the path of directed edges from variable $i$ to variable $j$ can only be satisfied if $x_{j} = 0$. Conversely, let us set $x_{i} = 1$ and let us now travel ``upstream'' to variable $j$. The clauses along the path of directed edges from variable $j$ to variable $i$ can only be satisfied if $x_{j} = 1$. Hence the lemma statement follows. 
\end{proof}

\begin{remark} The algorithm in Theorem \ref{thm:parity} cannot be used to solve the {\rm NP}-complete problem of finding an assignment with fixed Hamming weight (particle number) since one can generally not appropriately change the particle number by flipping a polynomial set of designated variables. Similarly, the algorithm above cannot be used to count the number of solutions to a 2-SAT instance (which is a \#P-complete problem), although there is a way of enumerating all solutions with effort growing with the number of solutions, which uses similar techniques  as our algorithm \cite{feder}. 
\end{remark}

\section{Proof of Lemma \ref{lem:QMA1}}
\label{app:qma}

\begin{proof}
Containment in $\textsc{QMA}$ can be obtained by simply mapping the $n$-mode fermionic problem onto a $n$-qubit problem using e.g. the Jordan-Wigner transformation in Eq.~\eqref{eq:JW}. This means that each projector $\Pi_i$ using $k$ fermionic modes is mapped to a projector which is a sum of terms, each of which has a $k$-local part to which a string of Pauli $Z$s of some length is appended. Hence the \textsc{Fermionic} $k$-SAT problem maps onto a local Hamiltonian problem but some Pauli strings in the Hamiltonian have weight larger than $k$. However, one can still apply standard QMA-arguments for the proof verification as in Section IV in \cite{Morimae_2016} since this method only relies on the fact that the number of Pauli terms in the Hamiltonian is ${\rm poly}(n)$. 

To prove hardness, we can map the $\textsc{QMA}_1$-complete problem in \cite{GN:3SAT} onto a \textsc{Fermionic} $k=9$-SAT problem using the unitary ``assimilation mapping" described in Lemma 15 in \cite{HATH:ferm}. We take the 3-local circuit Hamiltonian, with its $3$-local projectors, on some $n$ qubits in the $\textsc{QMA}_1$-proof and view it as a Hamiltonian on $n$ qubits plus $n/2$ fermionic modes which gets mapped on a space of $3n/2$ fermionic modes and the Pauli operators of one qubit get mapped on products of two (Majorana) fermion operators out of three fermionic modes. Thus each 3-local projector becomes a fermionic projector $\Pi_i$ involving at most $9$ fermionic modes which conserves fermionic parity. In the YES case, if there is a $n$-qubit satisfying assignment for the \textsc{Quantum} 3-SAT problem, it can be rewritten as a $3n/2$-fermionic state which is a satisfying assignment for the \textsc{Fermionic} 9-SAT problem. In the NO case, if there is no approximate satisfying assignment for \textsc{Quantum} 3-SAT, imagine that there is a fermionic state $\ket{\Phi}$ for the \textsc{Fermionic} 9-SAT problem for which $\sum_i \bra{\Phi}\Pi_i \ket{\Phi}=\beta$. Then performing the inverse mapping and tracing over the additional fermionic registers in $\ket{\Phi}\bra{\Phi}$ to get the qubit reduced density matrix $\rho^{\rm qubit}$, this implies that $\sum_i {\rm Tr}\, \Pi_i^{\rm qubit} \rho^{\rm qubit}=\beta$ which by assumption implies that $\beta\geq 1/{\rm poly}(n)$. 
\end{proof}

\end{document}